\documentclass[%
reprint,
]{revtex4-2}

\usepackage{amsmath}
\usepackage{amssymb}
\usepackage{bm,bbm}
\usepackage{graphicx}
\usepackage{enumitem}
\usepackage[bookmarks=false]{hyperref}
\hypersetup{colorlinks=true,citecolor=blue,linkcolor=red,urlcolor=blue,pdfstartview=FitH,bookmarksopen=true}

\usepackage[T1]{fontenc}
\usepackage[osf,sc]{mathpazo}
\usepackage{amsthm}
\usepackage{physics}
\usepackage{dsfont}
\usepackage{algorithm}
\usepackage{algpseudocode}
\newtheorem{theorem}{Theorem}

\newtheorem{lemma}[theorem]{Lemma}

\DeclareMathOperator{\diag}{diag}

\DeclareMathOperator{\Sp}{Span}

\DeclareMathOperator{\supp}{Supp}

\newcommand{\bs}{\bm{s}}
\newcommand{\ee}{\mathrm{e}}
\newcommand{\ii}{\mathrm{i}}

\newcommand{\maxover}[1][]{\underset{#1}{\mathrm{max}}}

\newcommand{\subto}{\mathrm{~s.t.}}
\newcommand{\FA}{\mathrm{~~for~all~}}

\newcommand{\cC}{\mathcal{C}}
\newcommand{\cD}{\mathcal{D}}

\newcommand{\cH}{\mathcal{H}}

\newcommand{\cL}{\mathcal{L}}

\newcommand{\cN}{\mathcal{N}}

\newcommand{\cP}{\mathcal{P}}

\newcommand{\cS}{\mathcal{S}}

\newcommand{\cU}{\mathcal{U}}
\newcommand{\cV}{\mathcal{V}}

\newcommand{\I}{\mathds{1}}

\usepackage[normalem]{ulem}

\begin{document}
	
\title{Spectral bounds for the partial transpose}
\author{Wei-Jie Jiang}
\affiliation{Department of Physics, Shandong University, Jinan 250100, China}
\author{Jing-Tao Qiu}
\affiliation{Department of Physics, Shandong University, Jinan 250100, China}
\author{Xiao-Dong Yu}
\email{yuxiaodong@sdu.edu.cn}
\affiliation{Department of Physics, Shandong University, Jinan 250100, China}

\begin{abstract}
Among the various entanglement measures, the negativity stands out not
only for its clear physical meaning but also for being directly
computable from the spectrum of the partial transpose.
However, the negativity captures only the total weight of the negative
eigenvalues, whereas the finer structure of the negative spectrum remains
largely unexplored.
In this work, we fill this gap by introducing a hierarchical
generalization of the negativity, the Ky Fan $k$-negativity, and
developing a unified analytical framework for deriving spectral bounds
on the partial transpose.
To obtain the absolute bounds,
we reduce the maximization problem to a spectral graph optimization,
whose solution yields the exact bound for every $k$ through a single
cubic equation.
We further investigate these spectral bounds both analytically and
numerically under a fixed-purity constraint.
In particular, we solve the $k=1$ case completely and uncover a simple
underlying graph structure.
As two direct applications, we show that the Ky Fan $k$-negativity
robustly certifies genuine multilevel entanglement and that it converts
the $p_3$-PPT condition into a quantitative lower bound on the
negativity.
\end{abstract}

\maketitle

\section{Introduction}

Quantum entanglement is a defining nonclassical
resource in quantum information science and serves as a key resource
for quantum teleportation
\cite{BennettTeleportingunknownquantum1993,
BouwmeesterExperimentalQuantumTeleportation1997,
HuExperimentalHighDimensional2020},
device-independent quantum key distribution \cite{AcinDeviceIndependentSecurity2007,VaziraniFullyDeviceIndependent2014,ZapateroAdvancesdeviceindependent2023},
measurement-based quantum computation
\cite{RaussendorfOneWayQuantum2001,
RaussendorfMeasurementbasedquantum2003,
BriegelMeasurementBasedQuantum2009},
and quantum metrology \cite{GiovannettiQuantumenhancedmeasurements2004,GiovannettiQuantumMetrology2006,PezzeEntanglementNonlinearDynamics2009}.
In entanglement theory, partial transposition plays a fundamental role.
In particular, the positive partial transpose (PPT) criterion provides
one of the most widely used methods for detecting entanglement
\cite{PeresSeparabilityCriterionDensity1996,
HorodeckiSeparabilitymixedstates1996}.
Beyond its role as a mathematical criterion, partial transposition is
closely connected to operational tasks such as entanglement
distillation
\cite{BennettPurificationNoisyEntanglement1996,HorodeckiMixedStateEntanglement1998,PlenioLogarithmicNegativityFull2005,LamiNosecondlaw2023}
and cost
\cite{AudenaertEntanglementCostPositive2003,WangCostQuantumEntanglement2020,GourDynamicalEntanglement2020,LamiComputableEntanglementCost2025,WangComputableFaithfulLower2025}.
Indeed, whether nondistillable states with a nonpositive partial
transpose exist remains one of the most prominent open problems in
entanglement theory
\cite{HorodeckiMixedStateEntanglement1998,
HorodeckiQuantumentanglement2009,
HorodeckiFiveOpenProblems2022}.
The applications of partial transposition also extend beyond
entanglement theory, for example, to condensed matter physics
\cite{ShapourianManyBodyTopological2017}, quantum field theory
\cite{CalabreseEntanglementNegativityQuantum2012,
RuggieroNegativityspectrumone2016}, and computer science
\cite{DiVincenzoQuantumdatahiding2002}.

While the nonpositivity of the partial transpose already certifies
entanglement, its spectral properties encode substantially richer
physical information. The most prominent spectral
quantity is the entanglement negativity
\cite{VidalComputablemeasureentanglement2002,
PlenioLogarithmicNegativityFull2005},
which measures the total magnitude of the negative eigenvalues and is
widely used to quantify entanglement. At the level of individual
eigenvalues, the smallest eigenvalue provides a complementary
perspective. Unlike the negativity, which can increase without bound
as the local dimension grows, the smallest eigenvalue is universally
bounded from below by $-1/2$
\cite{SanperaLocaldescriptionquantum1998,
RanaNegativeeigenvaluespartial2013}. This lower bound is saturated by
maximally entangled states of Schmidt rank two, including those
embedded in higher-dimensional Hilbert spaces. Consequently, large
negativity in high-dimensional systems cannot arise from a single
increasingly negative eigenvalue, but instead typically reflects the
collective contribution of multiple negative eigenvalues.

The spectral structure of the partial transpose is important not only
from a theoretical perspective, for example, for characterizing
entanglement witnesses
\cite{JohnstonNonpositivepartial2013,
Johnstoninverseeigenvalueproblem2018,
ShenInertiasentanglementwitnesses2020,
LiangInertiapartialtranspose2024},
but also for practical applications, particularly given the recent
interest in entanglement certification from moments of the partial
transpose \cite{ElbenMixedStateEntanglement2020,YuOptimalEntanglementCertification2021,NevenSymmetryresolvedentanglement2021}. However, general results on the spectrum of the
partial transpose remain scarce beyond the two extreme cases discussed
above. This scarcity can be attributed primarily to two factors: the
intrinsic complexity of the problem and the lack of theoretical tools
tailored to the intermediate spectral regime.

In this work, we address this gap by introducing the Ky Fan
$k$-negativity, denoted by $\cN_k$. This construction is conceptually
inspired by majorization theory
\cite{MarshallInequalitiesTheoryMajorization2011}, which characterizes
the concentration of probability distributions. The name reflects its
structural similarity to the Ky Fan $k$-norm, defined as the sum of the
$k$ largest singular values of an operator. By summing the absolute
values of the $k$ most negative eigenvalues of the partial transpose,
the Ky Fan $k$-negativity interpolates between the magnitude of its
minimum eigenvalue and the negativity.

The central question we address is: What are the absolute bounds on the
Ky Fan $k$-negativity? We answer this question from a graph-theoretic
perspective. We prove that the maximization of $\cN_k$ reduces to a
problem in spectral graph theory, whose solution gives the exact bound
$\cN_k^{\max}$ for every $k$. We further show that this maximum can be
attained only by pure states.

Beyond the absolute bounds attained by pure states, we generalize the
problem to mixed states and address the same question under a fixed
purity constraint, $p=\Tr(\rho^2)$. This constraint is experimentally
motivated, as the purity is one of the few global properties of a
quantum state that can be accessed without full state tomography \cite{EkertDirectEstimationsLinear2002,ElbenRenyiEntropiesRandom2018,ElbenMixedStateEntanglement2020,ZhouSingleCopiesEstimation2020}.
In this setting, we completely solve the fundamental case $k=1$.
For the general case $k>1$, we develop an efficient numerical algorithm
and derive analytical upper bounds.

Our results also have direct applications. As two examples, we show that
the Ky Fan $k$-negativity robustly certifies genuine multilevel entanglement
\cite{KraftCharacterizingGenuineMultilevel2018,
CongWitnessingIrreducibleDimension2017}
and upgrades the $p_3$-PPT condition
\cite{ElbenMixedStateEntanglement2020}
from a qualitative entanglement criterion to a quantitative lower
bound on the negativity.

Our paper is organized as follows. In Sec.~\ref{sec:prelim}, we define
the Ky Fan $k$-negativity and establish its fundamental properties. In
Sec.~\ref{sec:pure_k}, we derive its absolute bound and prove that this
bound can be attained only by pure states. In
Sec.~\ref{sec:mixed_k1}, we solve the purity-constrained problem for
$k=1$, while in Sec.~\ref{sec:general_k}, we address the general
mixed-state case for $k>1$. In Sec.~\ref{sec:app}, we present
applications of the Ky Fan $k$-negativity. Finally, we summarize our
results and discuss their implications in
Sec.~\ref{sec:conclusion}.

\section{Ky Fan negativity}\label{sec:prelim}
	
We consider a bipartite Hilbert space $\cH_A\otimes\cH_B$. For notational
simplicity, we assume that $\dim\cH_A=\dim\cH_B=d$, although all results
generalize directly to the case of unequal local dimensions. A quantum state is
represented by a positive semidefinite operator $\rho\succeq0$ satisfying
$\Tr(\rho)=1$.
For any state $\rho=\sum_{i,j,k,l=1}^d\rho_{ij;k\ell}\ketbra{ij}{kl}$,
the partial transpose on $\mathcal{H}_B$ is defined as
$\rho^\Gamma=\sum\rho_{ij;k\ell}\ketbra{i}{k}\otimes\ketbra{\ell}{j}
=\sum\rho_{i\ell;kj}\ketbra{ij}{kl}$.
Similarly, one can define the partial transposition on $\cH_A$, which differs 
from $\rho^\Gamma$ only by a global transposition and thus shares the exact 
same spectrum.

For any $k\in\mathds{N}$, we define the Ky Fan $k$-negativity as the
sum of the absolute values of the $k$ smallest negative eigenvalues of
$\rho^\Gamma$. Formally,
\begin{equation}\label{eq:def_kneg}
    \mathcal{N}_k(\rho) := 
	\sum_{\substack{1\le i\le k\\
	\lambda_i^{\uparrow}(\rho^\Gamma)<0}} |\lambda_i^{\uparrow}(\rho^\Gamma)|
\end{equation}
where $\lambda_i^{\uparrow}(\rho^\Gamma)$ are the negative eigenvalues of
$\rho^\Gamma$ in non-decreasing order. When $\rho$ is PPT, i.e.,
$\rho^\Gamma\succeq0$, all $\cN_k(\rho)$ are defined to be zero.

For $k=1$, $\mathcal{N}_1(\rho)$
is precisely the magnitude of the most negative eigenvalue
$|\lambda_{\min}(\rho^\Gamma)|$. As $k$ increases, $\cN_k$ accumulates
contributions deeper in the negative spectrum; when $k$ equals the total number
of negative eigenvalues, $\mathcal{N}_k(\rho)$ exactly recovers the standard
negativity $\mathcal{N}(\rho)$.

Although $\cN_k$ is not an entanglement measure, it shares many important
properties with entanglement measures. For example, $\cN_k(\rho)$ is invariant under local
unitary transformations. Another crucial property of $\cN_k(\rho)$ is its
invariance under the addition of local pure ancillae, i.e.,
\begin{equation}
    \cN_k(\ketbra{\psi_{A'}}\otimes\rho_{AB}\otimes\ketbra{\varphi_{B'}})
    =\cN_k(\rho_{AB}),
    \label{eq:ancilla}
\end{equation}
where $\ket{\psi_{A'}}$ and $\ket{\varphi_{B'}}$ are pure states on auxiliary
systems $\cH_{A'}$ and $\cH_{B'}$, respectively, 
and the partial transposition is taken with respect to the bipartition $A'A|BB'$.
Remarkably, the dimensions of $\cH_{A'}$ and $\cH_{B'}$ are not limited in
Eq.~\eqref{eq:ancilla}.
Physically, this implies that the dimensions of $\cH_A$ and $\cH_B$ can be 
considered arbitrarily large, with $\rho_{AB}$ embedded within a subspace of 
$\cH_A\otimes\cH_B$. Hereafter, we will always assume that the dimensions
of $\cH_A$ and $\cH_B$ are arbitrarily large unless stated otherwise.
Under this convention, the expression for $\cN_k(\rho)$ can
be simplified as
\begin{equation}\label{eq:kneg_nodim}
	\cN_k(\rho) = -\sum_{i=1}^k \lambda_i^{\uparrow}(\rho^\Gamma),
\end{equation}
since the number of zero eigenvalues separating the negative and positive
spectrum can be arbitrarily extended.
Equation~\eqref{eq:kneg_nodim} directly implies that $\cN_k(\rho)$ is a
convex function \cite{BoydConvexOptimization2004}.
In addition, we note that the equality \eqref{eq:ancilla} no longer holds 
when the ancillae are mixed states. This fact explains why $\cN_k$ can increase
under LOCC, and thus does not constitute a conventional entanglement measure
\cite{VidalEntanglementmonotones2000}.
We will later demonstrate that this non-monotonic property plays a
central role in applications.

\section{Absolute bound for Ky Fan negativity}\label{sec:pure_k}

In this section, we study the absolute bounds of $\cN_k(\rho)$. The lower bound
$\cN_k(\rho)\ge 0$ is obvious, and it is obtained if and only if $\rho$
is PPT. Thus, we focus on the upper bound, which extends the well-known
bound $\cN_1(\rho)\le\frac{1}{2}$ to all $\cN_k(\rho)$, providing a general 
constraint on the spectrum of the partial transpose.

Due to the convexity of $\cN_k(\rho)$, the maximization can always be
obtained on the pure states.
Let us express the pure state in its Schmidt decomposition as
$\ket{\psi}=\sum_{i=1}^d s_i\ket{ii}$, where $s_1\ge s_2\ge\cdots\ge s_d\ge 0$
and $\sum_{i=1}^d s_i^2=1$.
The negative eigenvalues of $\rho^\Gamma$ are then given by $-s_is_j$ for all
pairs $i<j$. 
Therefore, $\cN_k(\ket{\psi})$ can be expressed as
\begin{equation}
    \cN_k(\ket{\psi}) = \sum_{\text{$k$ largest}} s_i s_j,
\end{equation}
where the summation runs over the $k$ largest elements of the set
$\{s_i s_j\}_{i<j}$.

We start by examining a few concrete examples with small values of $k$.
The $k=1$ case is trivial, i.e., $s_1s_2\le (s_1^2+s_2^2)/2\le 1/2$.
For $k=2$, we need to select the two largest elements of $\{s_i s_j\}_{i<j}$.
As $s_1\ge s_2\ge\cdots\ge s_d\ge 0$, the largest two ones are obviously
$s_1s_2$ and $s_1s_3$.
Thus, we need to maximize $s_1 s_2 + s_1 s_3$ subject to the normalization constraint 
$s_1^2 + s_2^2 + s_3^2 = 1$ and the positivity condition $s_i \ge 0$
\footnote{In principle, it should be $s_1^2+s_2^2+s_3^2\le 1$,
but the maximum is clearly achieved when the equality holds.}.
Now observe that we can rewrite the objective function in a quadratic form:
\begin{equation}
	s_1s_2+s_1s_3
	=\frac12\begin{bmatrix}s_1, & s_2, & s_3\end{bmatrix}
	\underbrace{\begin{bmatrix}
			0 & 1 & 1\\
			1 & 0 & 0\\
			1 & 0 & 0
	\end{bmatrix}}_{A}
	\begin{bmatrix}s_1\\ s_2\\ s_3\end{bmatrix}
	\equiv \frac12 \bm{s}^T A \bm{s}.
    \label{eq:N2}
\end{equation}
One can easily see that $s_1s_2+s_1s_3$ is bounded by the
largest eigenvalue of $A$, more precisely, $\lambda_{\max}(A)/2$.
Moreover, this bound is achievable if the eigenvector corresponding to the 
largest eigenvalue can be chosen to be non-negative. This fact can be verified
by direct calculation, where the largest eigenvalue of $A$ is $\sqrt{2}$
and the corresponding eigenvector is $\bm{s}=(\sqrt{2}/2,1/2,1/2)$.
Thus, we obtain a universal bound for $\cN_2(\rho)$:
\begin{equation}
    \cN_2(\rho)\le\frac{\sqrt{2}}{2},
\end{equation}
and the bound is achieved if and only if $\rho$ is equal to
$(\sqrt{2}\ket{11}+\ket{22}+\ket{33})/2$ up to a local
unitary transformation. The necessity of this state configuration follows from
the result that the maximum of $\cN_k(\rho)$ is achieved only on pure states, a
fact we will prove later.
The fact that the corresponding eigenvector is non-negative is not a coincidence. 
This follows from the Perron-Frobenius theorem, which states that for any matrix 
with non-negative entries, the eigenvector corresponding to the largest 
eigenvalue can always be chosen to be non-negative \cite{PerronZurTheorieder1907,FrobeniusUbermatrizenaus1908}.

When $k \ge 3$, things become more complicated because the ordering of the 
products $s_is_j$ is no longer fixed. For example, whether $s_1s_4$ or 
$s_2s_3$ is the third largest value depends entirely on the specific choice 
of the vector $\bm{s}$. To obtain an analytical solution to the problem, 
we link it to a classical problem in spectral graph theory.

\begin{figure}
    \centering
    \includegraphics[width=\linewidth]{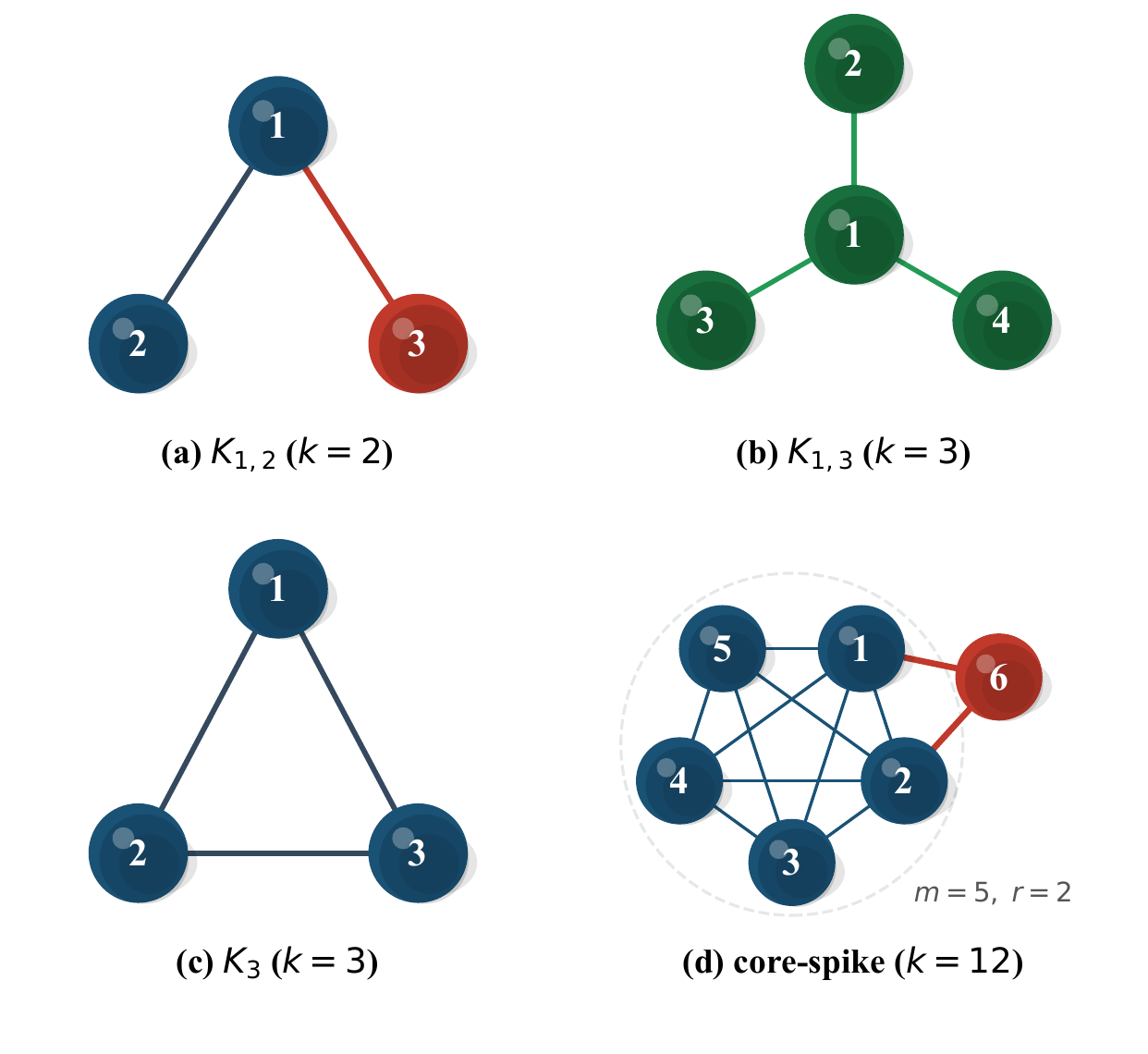}
    \caption{
        Four graphs with $k$ edges that are candidates for maximizing the 
        spectral radius. 
        (a)~The star graph $K_{1,2}$. 
        (b)~The star graph $K_{1,3}$. 
        (c)~The complete graph $K_3$. 
        (d)~The extremal graph $G^*_{12}$.
    }
    \label{fig:graph}
\end{figure}
The matrix $A$ in Eq.~\eqref{eq:N2} has a simple interpretation in graph theory:
it is the adjacency matrix of a
graph with vertices $\{1,2,3\}$ and edges $\{(1,2),(1,3)\}$, i.e., Fig.~\ref{fig:graph}(a).
In this graph, each active Schmidt coefficient corresponds to a
vertex, and each selected product $s_is_j$ corresponds to an edge between
vertices $i$ and $j$.  The quadratic form $\frac12\bm{s}^T A\bm{s}$
precisely sums the products $s_i s_j$ over all edges. 

For $k=3$ we must pick three products. Two possible candidates are the sets 
$\{s_1s_2, s_1s_3, s_1s_4\}$ and $\{s_1s_2, s_1s_3, s_2s_3\}$. They correspond
to graphs with three edges: a star graph $K_{1,3}$ with edges
$\{(1,2),(1,3),(1,4)\}$ (Fig.~\ref{fig:graph}(b)) and a triangle graph $K_3$ with edges
$\{(1,2),(1,3),(2,3)\}$ (Fig.~\ref{fig:graph}(c)). Their adjacency matrices are
\begin{equation}
    A_{\mathrm{star}}= 
\begin{bmatrix}
	0&1&1&1\\ 1&0&0&0\\ 1&0&0&0\\ 1&0&0&0
\end{bmatrix},\quad
A_{\mathrm{tri}}= 
\begin{bmatrix}
	0&1&1\\ 1&0&1\\ 1&1&0
\end{bmatrix}.
\end{equation}
Computing their largest eigenvalues gives
$\lambda_{\max}(A_{\mathrm{star}})=\sqrt{3}$,
$\lambda_{\max}(A_{\mathrm{tri}})=2$. The triangle graph therefore yields a
larger Ky Fan $3$-negativity. Thus, $\cN_3^{\max}=1$ and it is achieved when
$s_1=s_2=s_3=1/\sqrt{3}$.
In graph theory, the largest eigenvalue of the adjacency matrix is 
called the spectral radius of the graph.
This simple example illustrates that for a fixed number of edges, the
graph with denser connectivity has a larger spectral radius.

The same reasoning extends to arbitrary $k$. To determine the maximal $\cN_k$, we select a set $E$ of $k$ pairs $(i,j)$ with $i<j$ that will enter the sum. Identifying each $i$ with a vertex, and each selected pair with an edge, we obtain a graph $G=(V,E)$ with $|V|=n$ (the number of non-zero Schmidt coefficients) and $|E|=k$. The objective function takes the quadratic form
\begin{equation}
    \sum_{(i,j)\in E} s_i s_j = \frac12 \bm{s}^T A_G \bm{s},
\end{equation}
where $\bm{s}=(s_1,\dots,s_n)^T$ satisfies $\lVert\bm{s}\rVert=1$,
$s_i\ge 0$, and $A_G$ is the adjacency matrix of $G$. 
The maximum of this quadratic form for a fixed $G$ is $\frac12\lambda_{\max}(A_G)$. Therefore, the maximum over all states is
\begin{equation}
    \cN_k^{\max} = \frac12 \max_{\substack{G=(V,E)\\ |E|=k}} \lambda_{\max}(A_G).
    \label{eq:Nkmax}
\end{equation}
Then the problem of finding the exact upper bound for the Ky Fan $k$-negativity reduces to a question in spectral graph theory: among all graphs with
$k$ edges, which one attains the largest spectral radius?

The problem of maximizing the spectral radius of a graph given a fixed number
of edges is a classic challenge in spectral graph theory.
The problem was first posed by Brualdi and Hoffman in 1976 \cite[page
438]{BermondProblemesCombinatoireset1978}. They proved that for
$k=\binom{m}{2}$ the maximum spectral radius is attained by the
complete graph and conjectured that for general $k=\binom{m}{2}+r$ the maximum
spectral radius is attained by adding a new vertex and $r$ new edges to the
complete graph
\cite{Brualdispectralradius011985}; see an illustration in Fig.~\ref{fig:graph}(d).
Following the analytical bounds established by Stanley
\cite{Stanleyboundspectralradius1987} and Friedland
\cite{FriedlandBoundsspectralradius1988}, Rowlinson ultimately
confirmed the Brualdi-Hoffman conjecture \cite{Rowlinsonmaximalindexgraphs1988}.
Concretely, any integer $k$ can be uniquely decomposed as
\begin{equation}
    k = \binom{m}{2} + r, \qquad 0 \le r < m,
\end{equation}
where $m$ is the number of vertices that can be fully connected.
The extremal graph $G_k^*$ is then a core-spike structure: a complete
graph $K_m$ on $m$ vertices (the core), plus one additional vertex connected to
exactly $r$ of the core vertices (the spike). When $r=0$, the spike is absent
and $G_k^*$ is simply $K_m$. 

With the structure of the extremal graph identified, we can compute
$\lambda_{\max}(A_{G_k^*})$ exactly.  We can partition the vertex set of the
graph into three sets:
the set $\cS_1$ of $r$ core vertices connected to the spike, the set
$\cS_2$ of $m-r$ core vertices not connected to the spike, and the set $\cS_3$
of the spike vertex. This partition is equitable: every vertex in a given set
has the same number of neighbors in each of the three sets
\cite{GodsilAlgebraicgraphtheory2001}. We can therefore
record the number of neighbors between sets in a $3\times3$ quotient matrix
\begin{equation}
Q = \begin{bmatrix}
    r-1 & m-r & 1 \\
    r & m-r-1 & 0 \\
    r & 0 & 0
\end{bmatrix},
\end{equation}
where $Q_{ij}$ is the number of edges from a vertex in set $\cS_i$ to vertices
in set $\cS_j$, and the largest eigenvalue $\lambda_{\max}$ of $A_{G_k^*}$ coincides
with the largest eigenvalue of $Q$.
By noting that the characteristic polynomial of $Q$ is
\begin{equation}\label{eq:cubic_poly}
\lambda^3 - (m-2)\lambda^2 - (m+r-1)\lambda + r(m-r-1),
\end{equation}
we get the exact bound of Ky Fan $k$-negativity.

\begin{theorem}\label{thm:pure_exact}
    Let $k = \binom{m}{2} + r$ with $0 \le r < m$. Then the maximum Ky
    Fan $k$-negativity $\mathcal{N}_k^{\max}$ is given by the largest real root of
    \begin{equation}
       f(\lambda):=8\lambda^3 - 4(m-2)\lambda^2 - 2(m+r-1)\lambda + r(m-r-1).
    \end{equation}
\end{theorem}

One can easily verify that when $k=1,2,3$, the bounds $\cN^{\max}_1 = 1/2$, 
$\cN^{\max}_2 = 1/\sqrt{2}$, and $\cN^{\max}_3 = 1$ are recovered.
For the special case $r=0$, i.e., when $k=\binom{m}{2}$, one obtains
a simple result $\mathcal{N}_{\binom{m}{2}}^{\max} = \frac{m-1}{2}$.
The exact values of $\cN_k^{\max}$ obtained from Theorem~\ref{thm:pure_exact}
are plotted in Fig.~\ref{fig:Nk_max} for $k$ up to $36$. The continuous curve
shows Stanley's upper envelope $(\sqrt{1+8k}-1)/4$, which is attained precisely
when $k$ is a triangular number, i.e., $k=\binom{m}{2}$
\cite{Stanleyboundspectralradius1987}.
\begin{figure}
    \centering
    \includegraphics[width=0.9\columnwidth]{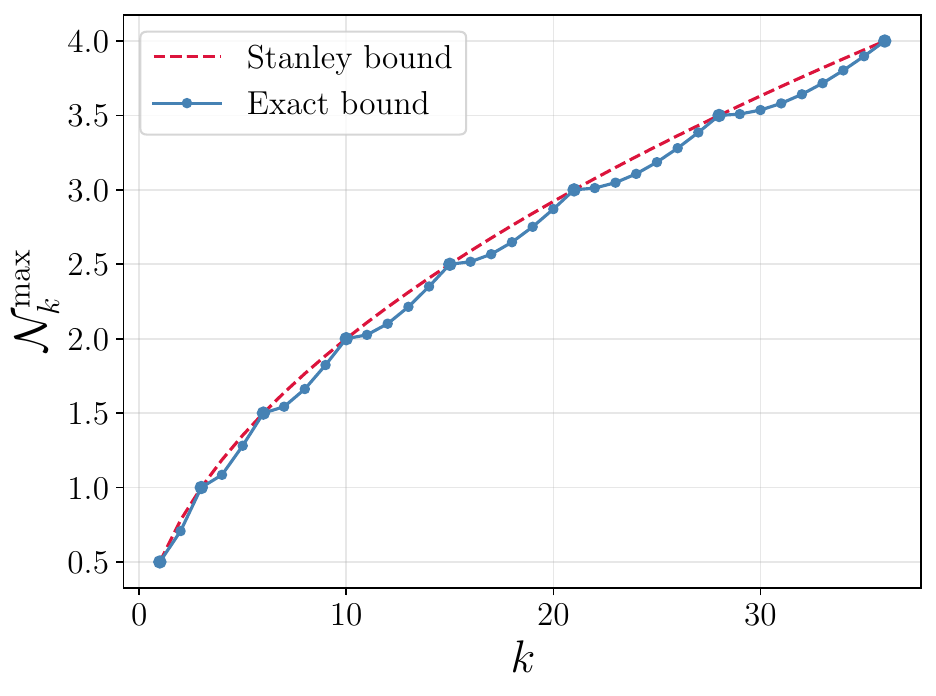}
    \caption{Exact maximum $\cN_k^{\max}$ for pure states as a function 
        of $k$.  The dashed red curve is Stanley's continuous envelope 
        $(\sqrt{1+8k}-1)/4$.}
    \label{fig:Nk_max}
\end{figure}

The optimal Schmidt coefficients that achieve $\mathcal{N}_k^{\max}$ are given
by the eigenvector of the quotient matrix $Q$. Let $\bm{s} = (s_1, s_2,
s_3)^T$ be the eigenvector corresponding to $\lambda_{\max}$, normalized such
that the full state vector has unit norm. The Schmidt decomposition of the
extremal state takes the form
\begin{equation}
    \ket{\psi_k^*} =  s_1 \sum_{i\in\cS_1} \ket{ii} + s_2 \sum_{j\in\cS_2}
    \ket{jj} + s_3 \ket{m+1,m+1},
\end{equation}
with the normalization $r s_1^2 + (m-r) s_2^2 + s_3^2 = 1$. For all $k>0$,
the inequality $s_1 \ge s_2 \ge s_3$ holds, reflecting the intuition that the
core vertices adjacent to the spike are the entanglement hub of the
configuration. In the limit $r=0$ (the spike disappears), $s_2$ and $s_3$
approach $s_1$ and the state becomes a uniform superposition over $m$ states,
recovering the well‑known maximally entangled state of dimension $m$.

Before proceeding to the mixed-state case, we show that the upper bound 
$\cN_k^{\max}$ can only be achieved by pure states. 
By the convexity of $\cN_k(\rho)$, all pure states $\ket{\psi_i}$ 
in the decomposition $\rho=\sum_{i}p_i\ketbra{\psi_i}{\psi_i}$ 
must achieve the optimal value $\cN_k^{\max}$. The unitary freedom of 
decomposition further implies that any state $\ket{\psi}$ within the 
support of $\rho$ must also be an optimal pure state, possessing the 
identical Schmidt coefficients dictated by the optimal graph $G_k^*$.
For any two pure states $\ket{\psi_1}$ and $\ket{\psi_2}$ in the support of
$\rho$, any normalized superposition 
\begin{equation}
    \ket{\phi(\theta, \varphi)} = \cos\theta\ket{\psi_1} + \sin\theta \ee^{\ii\varphi}\ket{\psi_2}
\end{equation}
must also share the same Schmidt coefficients. 
This forces all moments of the reduced density matrix of $\ket{\phi(\theta,
\varphi)}$ to be phase-independent constants. A Fourier projection over the
superposition phase explicitly reveals that $X = \Tr_B(\ketbra{\psi_1}{\psi_2})$
must be strictly nilpotent, 
but expanding the states $\ket{\psi_1}$ and $\ket{\psi_2}$ within identical
basis demonstrates that $X$ is inherently an invertible matrix. 
The conflict between nilpotency and invertibility proves that the global maximum
is exclusively achieved by pure states. Please see
Appendix~\ref{app:pure_state_proof} for more details.
	
\section{Mixed states: the case $k=1$ as an illustrative example}
\label{sec:mixed_k1}

Having fully characterized the bounds for pure states, we now direct our
attention to mixed states $\rho$ subject to a purity constraint $p =
\Tr(\rho^2)$.  As the most fundamental scenario, we completely resolve the $k=1$
case. We find the exact analytical upper bound of
\begin{equation}
\cN_1^{\max}(p)=\max_{\Tr(\rho^2)=p}\cN_1(\rho)
=-\min_{\Tr(\rho^2)=p}\lambda_{\min}(\rho^\Gamma).
\end{equation}
This section demonstrates how the graph-theoretic
method extends to mixed states and yields explicit purity-dependent bounds.
	
We begin our analysis with the variational definition of the bound
\begin{equation}
\begin{aligned}
    \cN_1^{\max} (p)&= \max_{\Tr(\rho^2)=p}\max_{\ket{\psi}}
    \Tr(-\rho^\Gamma \ketbra{\psi}) \\
    &= \max_{\Tr(\rho^2)=p}\max_{\ket{\psi}}
    \Tr\bigl(-\rho(\ketbra{\psi})^\Gamma\bigr).
\end{aligned}	
\label{eq:N1max}
\end{equation}
Let us express the optimal test state in its Schmidt decomposition as
$\ket{\psi}=\sum_{i=1}^d s_i\ket{ii}$, where $s_1\ge s_2\ge\cdots\ge s_d\ge 0$
and $\sum_{i=1}^d s_i^2=1$.
$\ketbra{\psi}^\Gamma$ possesses negative eigenvalues  $-s_i s_j$ with
corresponding eigenvector $\ket*{\psi_{ij}^-} = (\ket{ij} - \ket{ji})/\sqrt{2}$
for $i < j$. 
Note that we do not restrict the dimension of the systems, thus
the optimal $\rho$ must be entirely supported on the antisymmetric subspace
spanned by $\{\ket*{\psi_{ij}^-}\}$, i.e.,
\begin{equation}\label{eq:rho_antisym_mixed}
    \rho = \sum_{i<j} m_{ij} \ketbra*{\psi_{ij}^-}
\end{equation}
with $m_{ij} \ge 0$, $\sum_{i<j} m_{ij} = 1$, and $\sum_{i<j} m_{ij}^2 = p$
\footnote{In principle, it should be $\rho = \sum_{i<j,k<l} m_{ij,kl}
\ketbra*{\psi_{ij}^-}{\psi_{kl}^-}$, but the optimality can always
be achieved when $\rho$ is diagonal under the basis
$\{\ket*{\psi_{ij}^-}\}_{i<j}$.}.
Substituting Eq.~\eqref{eq:rho_antisym_mixed} into Eq.~\eqref{eq:N1max}, we
obtain that

\begin{equation}
    \cN_1^{\max}(p)
    =\max_{m_{ij}}\max_{s_i}\sum_{i<j}s_is_jm_{ij}.
\end{equation}
Defining a symmetric non-negative matrix $M$ with $M_{ij}=m_{ij}$ for $i\neq j$
and $M_{ii}=0$, the constraints become $\sum_{i\neq j} M_{ij} = 2$ and
$\sum_{i\neq j} M_{ij}^2 = 2p$. Consequently, the bound $\cN_1^{\max}(p)$
reduces to a nonlinear optimization problem:
\begin{equation}
\begin{aligned}
&\maxover[M,\bs] \quad && \frac12\sum_{i\ne j} M_{ij}s_is_j\\
&\subto && \sum_is_i^2=1,~
\sum_{i\ne j}M_{ij}=2,\\
&&&\sum_{i\ne j}M_{ij}^2=2p,~M_{ij}\ge 0.
\end{aligned}
\label{eq:N1opt}
\end{equation}
Here, we omit the positivity constraints $s_i\ge 0$, because
similar to Eq.~\eqref{eq:Nkmax}, $\cN_1^{\max}(p)$ can also be reformulated
as an eigenvalue problem
\begin{equation}
    \cN_1^{\max}(p) = \max_{M,\bs} \frac12
    \bs^T M \bs = \frac12 \max_{M}\lambda_{\max}(M),
\end{equation}
and the positivity constraints $s_i\ge 0$ can be guaranteed by the
Perron-Frobenius theorem.

Similar to the pure state scenario, the matrix $M$ also has a graph theory
interpretation. We can map every such $\rho$ to a weighted graph $G = (V, E,
M)$. In this graph, there is an edge between vertex $i$ and vertex $j$ if
$M_{ij}>0$. $M_{ij}$ serves as the weight of this edge. Then the matrix
$M$ is naturally the adjacency matrix of the graph.
Thus, the problem of determining $\cN_1^{\max}(p)$ is exactly recast as a
spectral graph theory problem: maximizing the spectral radius of a weighted
graph subject to fixed total weight and total squared weight.

To solve this problem analytically, we utilize the Karush-Kuhn-Tucker conditions
\cite{KuhnNonlinearprogramming2013,KarushMinimaFunctionsSeveral2014} to tackle
the nonlinear optimization in Eq.~\eqref{eq:N1opt}.
We construct the Lagrangian function with multipliers $\lambda$ (enforcing $\sum
s_i^2 = 1$), $\alpha$ (for the sum of edge weights $\sum_{i\neq j} M_{ij} = 2$),
$\beta$ (for the purity constraint $\sum_{i\neq j} M_{ij}^2 = 2p$), and
$\gamma_{ij} \ge 0$ (enforcing non-negativity $M_{ij} \ge 0$):
\begin{equation}
    \begin{aligned}
        \mathcal{L} =&\sum_{i\neq j} M_{ij} s_i s_j - \lambda\Bigl(\sum_i s_i^2 - 1\Bigr) - \alpha\Bigl(\sum_{i\neq j} M_{ij} - 2\Bigr)\\
        &- \beta\Bigl(\sum_{i\neq j} M_{ij}^2 - 2p\Bigr) + \sum_{i\neq j}\gamma_{ij}M_{ij}.
    \end{aligned}
\end{equation}
Since $M$ is symmetric, we regard $M_{ij}$ and $M_{ji}$ as the same variable
while retaining the notation $\sum_{i\neq j}$ for convenience and correspondingly take $\gamma_{ij}=\gamma_{ji}$.
Imposing the stationarity condition with respect to the edge weights $M_{ij}$ yields
\begin{equation}\label{eq:stationary_M}
    \frac{\partial \cL}{\partial M_{ij}}=2(s_i s_j - \alpha - 2\beta M_{ij} + \gamma_{ij}) = 0.
\end{equation}
If $\beta\neq 0$, combining Eq.~\eqref{eq:stationary_M} with the complementary slackness condition $\gamma_{ij} M_{ij} = 0$, we
obtain the element-wise rule
\begin{equation}\label{eq:elementwise}
    M_{ij} = \max\bigl(a s_i s_j + b, 0\bigr),
\end{equation}
where we define $a = 1/(2\beta) > 0$ and $b = -\alpha/(2\beta)$.
The sign of $\beta$ is easy to understand. If $\beta$ were
negative, the optimal solution of Eq.~\eqref{eq:elementwise}
would assign smaller weights to edges with larger
product $s_is_j$, which contradicts the objective of maximization.
Eq.~\eqref{eq:elementwise} defines a strict threshold rule: an edge exists
between vertex $i$ and vertex $j$ ($M_{ij} > 0$) if and only if the product of
their eigenvector components strictly exceeds the threshold $\alpha$:
\begin{equation}
    s_i s_j > -\frac{b}{a} = \alpha.
\end{equation}
If $\beta=0$, the stationarity condition reduces to $s_i s_j=\alpha$ on every edge $(i,j)$ and $s_is_j\le\alpha$ otherwise.

Solving for the maximal spectral radius of $M$ is therefore equivalent to uncovering the optimal edge topology that satisfies the threshold rule. The following structural lemma severely restricts the permissible graph.

\begin{lemma}\label{lem:structure_mixed}
    Under the optimality conditions, the non-zero components of the eigenvector $\bs$ can take at most two distinct values. If two distinct values occur, the smaller value appears at most once.
\end{lemma}

The main idea of the proof is as follows: our analysis shows that the non-zero components $s_i$ satisfy the same cubic polynomial equation, which restricts them to at most two distinct values. We then use perturbation arguments to show that, whenever two distinct values occur, the smaller value can appear only once. The details are shown
in Appendix~\ref{app:proof_structure}.

\begin{figure}
    \centering
    \includegraphics[width=\linewidth]{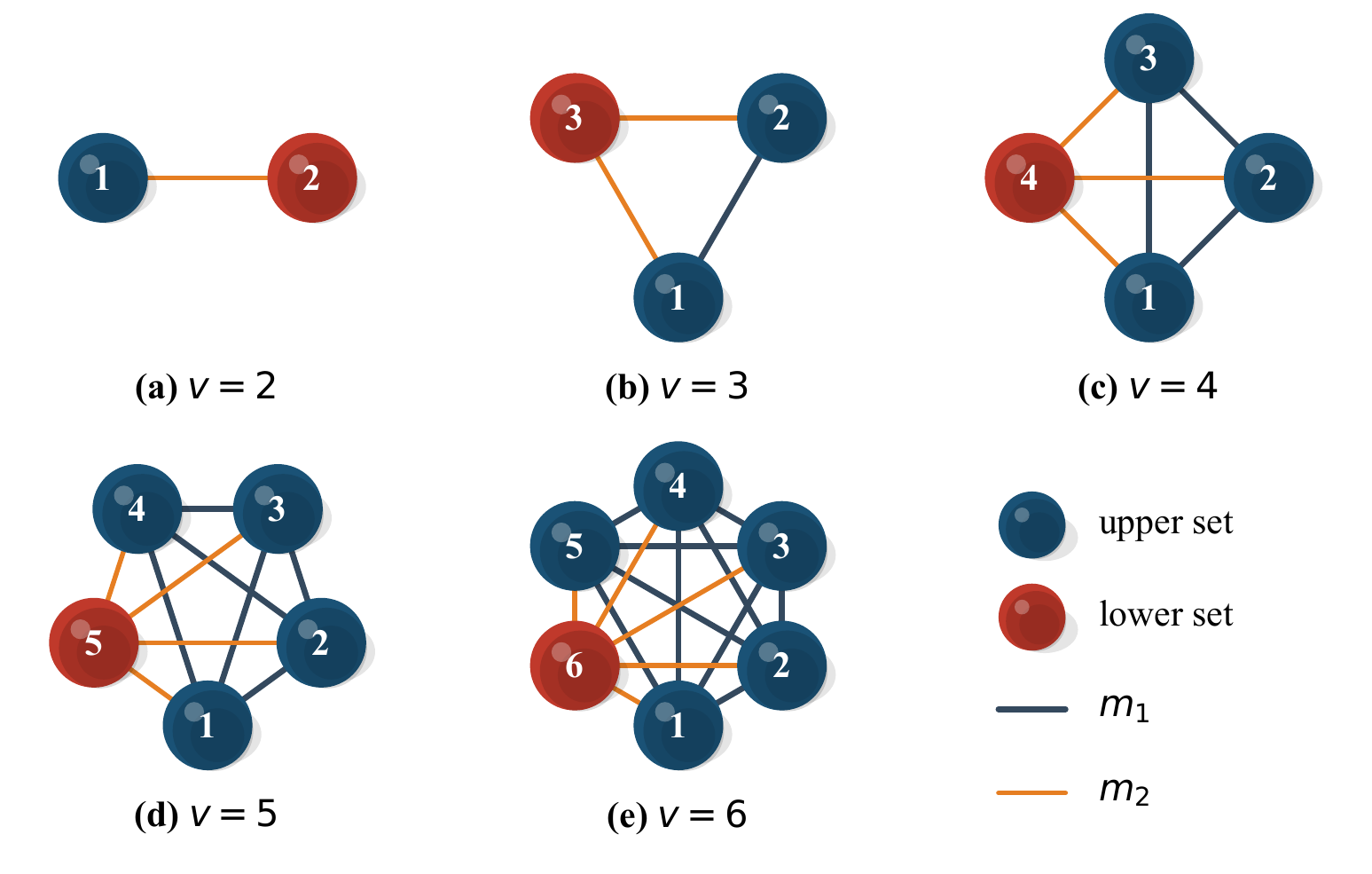}
    \caption{
        Weighted complete graphs that realize the optimal structure of
        Lemma~\ref{lem:structure_mixed} for vertex numbers $v$.
        The upper set ($v-1$ vertices, blue) is a complete clique with
        uniform edge weight $m_1$; the single lower vertex (red) is
        connected to every upper vertex with weight $m_2$ (orange).}
    \label{fig:weighted_graphs}
\end{figure}

Lemma~\ref{lem:structure_mixed} implies that the optimal graph is a complete
graph with $v$ vertices which is partitioned into an upper set of $v-1$ vertices
sharing a larger value, and a lower set of a single vertex carrying a smaller
value. Then the matrix $M$ has two distinct non-zero entries: $m_1$ on edges
within the upper set, and $m_2$ on edges between the upper set and the single
lower vertex. The optimal weighted-graphs are illustrated in Fig.~\ref{fig:weighted_graphs}. 
Let $n_1 = (v-1)(v-2)/2$ denote the number of edges in the upper set,
$n_2 = v-1$ denote the number of edges between the upper set and the lower set,
and $n = n_1 + n_2 = v(v-1)/2$ denote the total number of edges. Consequently,
the normalization and purity constraints reduce to 
$n_1 m_1 + n_2 m_2 = 1$ and
$n_1 m_1^2 + n_2 m_2^2 = p$, respectively. 
Solving the equations under $m_1 \ge m_2 > 0$ yields the edge weights:
\begin{align}
    \label{eq:m1}
    m_1 &= \frac{2}{v(v-1)}\qty[1 + \sqrt{\frac{pv^2-pv-2}{v-2}}], \\
    \label{eq:m2}
    m_2 &= \frac{2}{v(v-1)}\qty[1 - \sqrt{\frac{(v-2)(pv^2-pv-2)}{4}}].
\end{align}
The positivity of $m_2$ forces $p$ to lie in the interval
\begin{equation}
    \frac{2}{v(v-1)} \le p < \frac{2}{(v-1)(v-2)},
\end{equation}
which uniquely determines the integer $v$ for a given $p$,
\begin{equation}
    \label{eq:v}
    v =\left\lceil\frac{\sqrt{1+8p^{-1}}+1}{2} \right\rceil.
\end{equation}
The spectral radius of $M$ coincides exactly with the largest eigenvalue of the $2 \times 2$ reduced quotient matrix:
\begin{equation}
    Q = \begin{bmatrix}
        (v-2)m_1 & m_2 \\
        (v-1)m_2 & 0
    \end{bmatrix},
    \label{eq:Qmix}
\end{equation}
Thus, the analytical expression for $\cN_1^{\max}(p)$ can be derived from
$\frac12 \lambda_{\max}(M)$.

\begin{theorem}\label{thm:N1max}
For any quantum state with purity $\Tr(\rho^2)=p$, $\cN_1^{\max}(p)$ is
given by
\begin{equation}
    \cN_1^{\max}(p) = \frac{(v-2)m_1 + \sqrt{(v-2)^2 m_1^2 + 4(v-1)m_2^2}}{4},
\end{equation}
and the bound is obtained by
\begin{equation}
    \rho = m_1 \sum_{1 \le i < j \le v-1} \ketbra*{\psi_{ij}^-} + m_2
    \sum_{i=1}^{v-1} \ketbra*{\psi_{iv}^-},
\end{equation}
where $m_1$, $m_2$, $v$ are given by
Eqs.~(\ref{eq:m1},\,\ref{eq:m2},\,\ref{eq:v}), respectively.
\end{theorem}

As illustrated in Fig.~\ref{fig:bound_curve}, the exact bound manifests as a
piecewise function of the purity $p$. As the state becomes increasingly
mixed, the effective dimension $v$ required to sustain the maximal negativity
undergoes discrete upward jumps.
Analogous to Stanley's bound in the discrete case, we can construct a smooth
envelope for this piecewise bound $\cN_1^{\max}(p)$. The key point is to
observe that the non-smoothness originates from the discontinuity of $v$. By
dropping the ceiling function in Eq.~\eqref{eq:v}, we obtain a smooth upper
bound from Eq.~\eqref{eq:Qmix} (see Appendix~\ref{app:upper_bound} for the rigorous proof)
\begin{equation}
	\cN_1^{\max}(p) \le \frac{2}{1 + \sqrt{1 + 8p^{-1}}}.
\end{equation}
This relaxed bound, shown as a black dashed curve in Fig.~\ref{fig:bound_curve},
coincides with the exact piecewise bound precisely at the set of endpoints $p =
\frac{2}{v(v-1)}$.
In the low-purity limit, the bound is nearly tight and scales as
$\cN_1^{\max}(p) \approx \sqrt{p/2}$.

\begin{figure}
	\centering
	\includegraphics[width=0.45\textwidth]{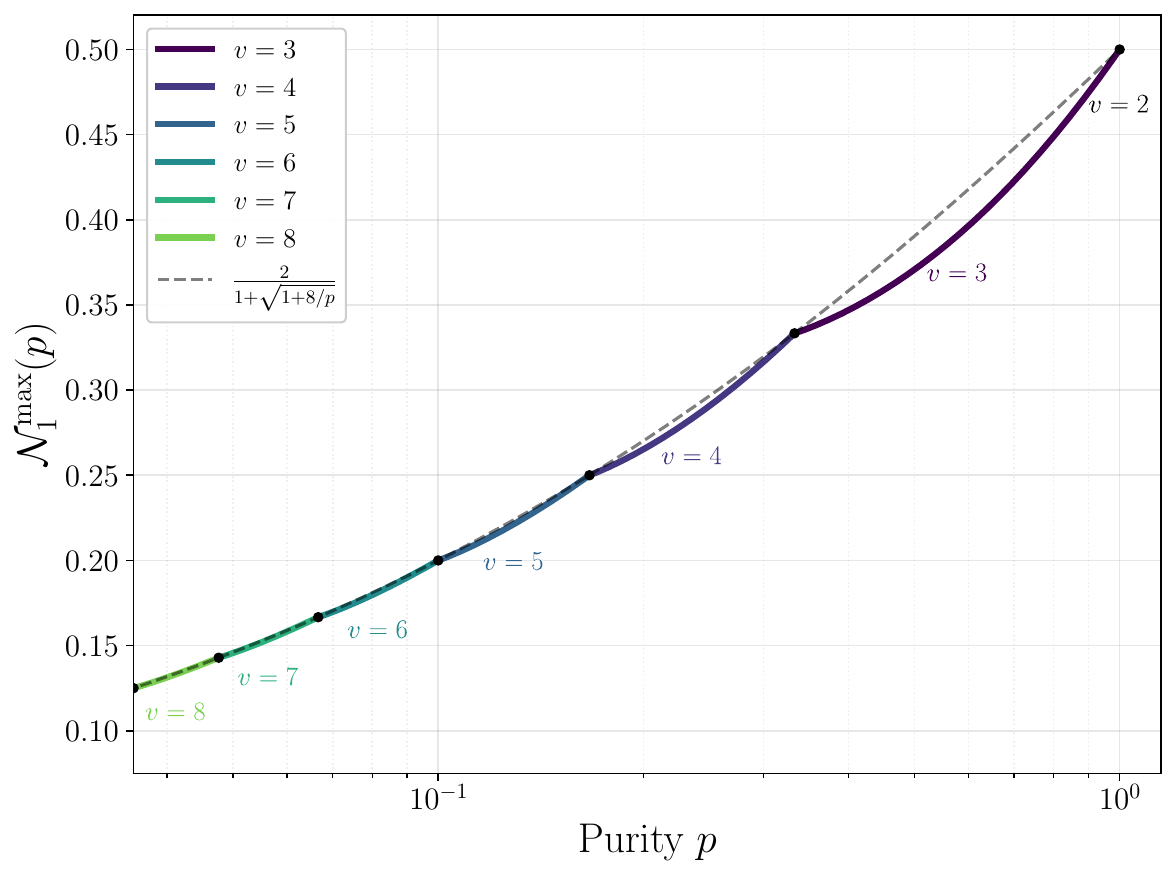}
	\caption{$\cN_1^{\max}(p)$ plotted against a logarithmic purity scale. The
	bound is a piecewise curve, where each distinctly colored segment corresponds
	to a specific effective local dimension $v$.}
	\label{fig:bound_curve}
\end{figure}

\section{Mixed states: analytical and numerical bounds for $k\ge 2$}
\label{sec:general_k}

In the preceding sections, we demonstrated that the maximization of the Ky Fan $k$-negativity can be elegantly reduced to a graph spectral optimization for pure states (for any $k$) and for mixed states in the fundamental $k=1$ regime. For general mixed states with $k>1$, however, this intuitive graph-theoretic mapping fundamentally breaks down.
To make progress, we proceed in two complementary directions: we first
derive rigorous analytical upper bounds on $\cN_k(\rho)$ under a purity
constraint, and then evaluate $\cN_k^{\max}(p)$ numerically. 

\begin{theorem}\label{thm:upper_bound}
Let $\rho$ be a quantum state with $\Tr(\rho^2)= p$. Then
\begin{equation}\label{eq:upper_bound}
    \cN_k(\rho)\le\sqrt{g_k p}\le\sqrt{kp},\quad
    g_k:=\sup_{\substack{P^2=P=P^\dagger,\\ \rank(P)=k}}\,\Tr\!\left[(P^\Gamma)_-^2\right],
\end{equation}
where the supremum is taken over rank-$k$ projectors $P$, and
$(P^\Gamma)_-:=(|P^\Gamma|-P^\Gamma)/2$ denotes the negative part of a Hermitian
operator $P^\Gamma$.
\end{theorem}

\begin{proof}
Let $\tilde P$ be the projector onto the subspace of the $k$ smallest
eigenvalues of $\rho^\Gamma$, so that $\cN_k(\rho)=\Tr(-\tilde P\rho^\Gamma)$
by Eq.~\eqref{eq:kneg_nodim}. Since the partial transpose is
self-adjoint with respect to the Hilbert-Schmidt inner product, we have
\begin{equation}
\begin{aligned}
    \cN_k(\rho)&=\Tr(-\tilde P^\Gamma\rho)\le\Tr[(\tilde P^\Gamma)_-\rho]\\
    &\le\|(\tilde P^\Gamma)_-\|_F\,\|\rho\|_F
    \le\sqrt{g_k p}\le\sqrt{kp}.
\end{aligned} 
\end{equation}
Here, the first inequality drops the nonnegative part of $\tilde P^\Gamma$, the second is the Cauchy-Schwarz inequality
together with $\|\rho\|_F=\sqrt{p}$, the third uses
the definition of $g_k$, and the last one follows from
$\Tr[(P^\Gamma)_-^2]\le\|P^\Gamma\|_F^2=k$.
\end{proof}

For every $k$, \(1/2\leq g_k/k\leq1\), and
\(\lim_{k\to\infty}g_k/k=1\). Consequently, the bound in
Eq.~\eqref{eq:upper_bound} scales as \(\sqrt{kp}\). In the
low-purity regime, this scaling is optimal up to at most a factor of
\(\sqrt{2}\), and the bound becomes asymptotically tight as
\(k\to\infty\). See Appendix~\ref{app:scaling} for more details.
The bound is also consistent with the
exact solution for $k=1$ in Sec.~\ref{sec:mixed_k1}: for a rank-one
projector with Schmidt coefficients $s_i$, one has
$\Tr[(P^\Gamma)_-^2]=\sum_{i<j}s_i^2s_j^2=(1-\sum_i s_i^4)/2$, which
approaches $g_1=1/2$ as the dimension grows, so that
Theorem~\ref{thm:upper_bound} reproduces the low-purity scaling
$\sqrt{p/2}$ of $\cN_1^{\max}(p)$.

Having established the analytical upper bounds, we now turn to the
numerical evaluation of $\cN_k^{\max}(p)$. The numerical methods provide
lower bounds on $\cN_k^{\max}(p)$. Together with the analytical upper
bounds derived above, they bracket $\cN_k^{\max}(p)$ from both sides.
Let us define a set,
\begin{equation}
\cS_k=\{P\mid 0\preceq P\preceq\I,~\Tr(P)\le k\},    
\end{equation}
which is the convex hull of all rank-$r$ orthonormal projectors with $r\le k$.
One can easily verify that
\begin{equation}
\cN_k(\rho)
=\max_{r\le k}\sum_{i=1}^r-\lambda_i^{\uparrow}(\rho^\Gamma)
=\max_{P\in\mathcal{S}_k}\Tr(-\rho^\Gamma P).
\label{eq:opt_fantope}
\end{equation}
Coupled with the purity constraint $\Tr(\rho^2)=p$, we face the joint
optimization
\begin{equation}
    \begin{aligned}
        \max_{\rho, P} \quad & \Tr(-\rho^\Gamma P) \\
        \subto \quad & \rho \succeq 0, \quad \Tr(\rho) = 1, \quad \Tr(\rho^2) = p, \\
        & 0 \preceq P \preceq \I, \quad \Tr(P) \le k.
    \end{aligned}
    \label{eq:optimization_problem}
\end{equation}
Before proceeding, we add the following remarks. First, according to 
Eq.~\eqref{eq:kneg_nodim}, it seems that we can use the standard Fantope 
$\{P\mid 0\preceq P\preceq\I,~\Tr(P)=k\}$
\cite{DatorroConvexOptimizationEuclidean2005}
in Eq.~\eqref{eq:opt_fantope}.
However, Eq.~\eqref{eq:kneg_nodim} holds only if the dimension of the 
quantum system can be sufficiently large.
In numerical optimization, this requires introducing additional dimensions,
resulting in increased computational cost.
Therefore, we introduce the new set $\cS_k$, which makes Eq.~\eqref{eq:opt_fantope}
hold in general. To attain the absolute global bound, we increase the dimension $d$
of the optimization search space. This dimensional expansion is continued until
the computed bound ceases to grow, indicating that the search space is
sufficiently large to capture the true bound of $\cN_k$.
Second, the purity equality constraint $\Tr(\rho^2) = p$ is not convex. 
Thus, we relax it to the convex inequality constraint $\Tr(\rho^2) \le p$, 
or equivalently, $\norm{\rho}_F \le \sqrt{p}$. One can anticipate that the 
maximum Ky Fan $k$-negativity $\cN_k^{\max}(p)$ is non-decreasing with 
respect to $p$; therefore, this relaxation does not alter the maximum value 
in Eq.~\eqref{eq:optimization_problem}. See Appendix~\ref{app:purity_relaxation}
for the rigorous proof.

To solve the optimization in Eq.~\eqref{eq:optimization_problem}, we employ a
heuristic seesaw algorithm (alternating optimization)
\cite{WernerBellInequalitiesEntanglement2001}. By freezing one variable at a
time, the bilinear objective function is decoupled into two tractable
sub-problems that can be solved iteratively. The algorithm proceeds as follows:

We embed the optimization in a $d \times d$ system, starting from a minimal
dimension. We run the inner loop to find the maximal value $\cN_k^{(d)}$ for the
current $d$. Then, we increment $d$ and repeat the process. The expansion
terminates when $\cN_k^{(d)}$ converges.
		
For a fixed dimension $d$, we initialize the algorithm with a randomly generated
valid rank-$k$ projector $P^{(0)}$. To mitigate the risk of converging to local
optima within the non-convex landscape, the entire seesaw procedure is repeated
across multiple random initializations.

First, given a fixed projector $P^{(n)}$, the problem reduces to finding the
optimal state $\rho^{(n)}$
\begin{equation}
    \begin{aligned}
        \rho^{(n)} = \arg\max_{\rho} \quad 
        & \Tr[-\rho (P^{(n)})^\Gamma] \\
        \textrm{s.t.} \quad & \rho \succeq 0, \quad \Tr(\rho) = 1,\\
        & \|\rho\|_F \le \sqrt{p}.
    \end{aligned}
\end{equation}
This is a convex optimization problem, which can be efficiently solved
\cite{BoydConvexOptimization2004}.

Second, given the updated state $\rho^{(n)}$, we seek to maximize
$\mathrm{Tr}(-(\rho^{(n)})^\Gamma P)$. The optimal rank-$k$ projector
$P^{(n+1)}$ that maximizes this trace is exactly the projector onto the
subspace spanned by the eigenvectors corresponding to the $k$ most negative
eigenvalues of $(\rho^{(n)})^\Gamma$; more precisely,
\begin{equation}
    P^{(n+1)} = \sum_{\substack{1\le i\le k\\ \lambda_i^{\uparrow}\le 0}}
    \ketbra{\psi_i}{\psi_i},
\end{equation}
where $(\rho^{(n)})^{\Gamma}=\sum_{i=1}^{d^2}\lambda_i^{\uparrow}
\ketbra{\psi_i}{\psi_i}$. 
The optimizations are iterated until the objective function value converges within a predefined numerical tolerance. 

The procedure is summarized in Algorithm~\ref{alg:seesaw}.

\begin{algorithm}[H]
    \caption{Two-layer Optimization for $\cN_k^{\max}(p)$}\label{alg:seesaw}
    \begin{algorithmic}[1]
        \Require Purity $p$, target $k$, tolerance $\epsilon$
        \State Initialize subsystem dimension $d = d_{\min}$, $\cN_k^{\rm prev} = 0$
        \Repeat \Comment{Outer loop: Dimension expansion}
        \State Initialize a random projector $P \in \cS_k$
        \Repeat \Comment{Inner loop: Seesaw algorithm}
        \State  $\rho \gets \arg\max_{\rho} \Tr(-\rho P^\Gamma)$ 
        
        \State Calculate $(\rho)^{\Gamma}=\sum_{i=1}^{d^2}\lambda_i^{\uparrow} \ketbra{\psi_i}{\psi_i}$
        \State  $P \gets \sum_{\substack{1\le i\le k\\ \lambda_i^{\uparrow}\le 0}} \ketbra{\psi_i}{\psi_i}$
        \Until{the objective function converges}
        \State Let $\cN_k^{(d)}$ be the converged maximum value
        \State $\Delta\cN \gets |\cN_k^{(d)} - \cN_k^{\rm prev}|$ \Comment{Evaluate saturation}
        \State $\cN_k^{\rm prev} \gets \cN_k^{(d)}$
        \State $d \gets d + 1$
        \Until{$\Delta\cN < \epsilon$}
        \State \Return $\cN_k^{\rm prev}$
    \end{algorithmic}
\end{algorithm}

We now apply Algorithm~\ref{alg:seesaw} to evaluate the bound of Ky Fan
$k$-negativity for $k \in \{1,2,3,4,5,6\}$ across the purity $p \in (0.1, 1]$. The
numerical results are plotted in Fig.~\ref{fig:purity_phase_diagram}.
It is important to emphasize that our heuristic approach yields only a
numerical estimate, which provides a lower bound on the true absolute
maximum rather than a mathematically certified optimum.
Nevertheless, the numerical results provide two independent and
compelling consistency checks, strongly suggesting that the obtained
lower bounds are optimal. First, the numerical
curve for $k=1$ perfectly overlays the exact bound $\cN_1^{\max}(p)$
established in Sec.~\ref{sec:mixed_k1}. Second, in the pure state limit $p=1$,
the numerical values for all $k \ge 1$ precisely hit the theoretical bounds
$\cN_k^{\max}$ derived in Sec.~\ref{sec:pure_k}. These exact alignments strongly
validate the algorithm's accuracy.

\begin{figure}
    \centering
    \includegraphics[width=\linewidth]{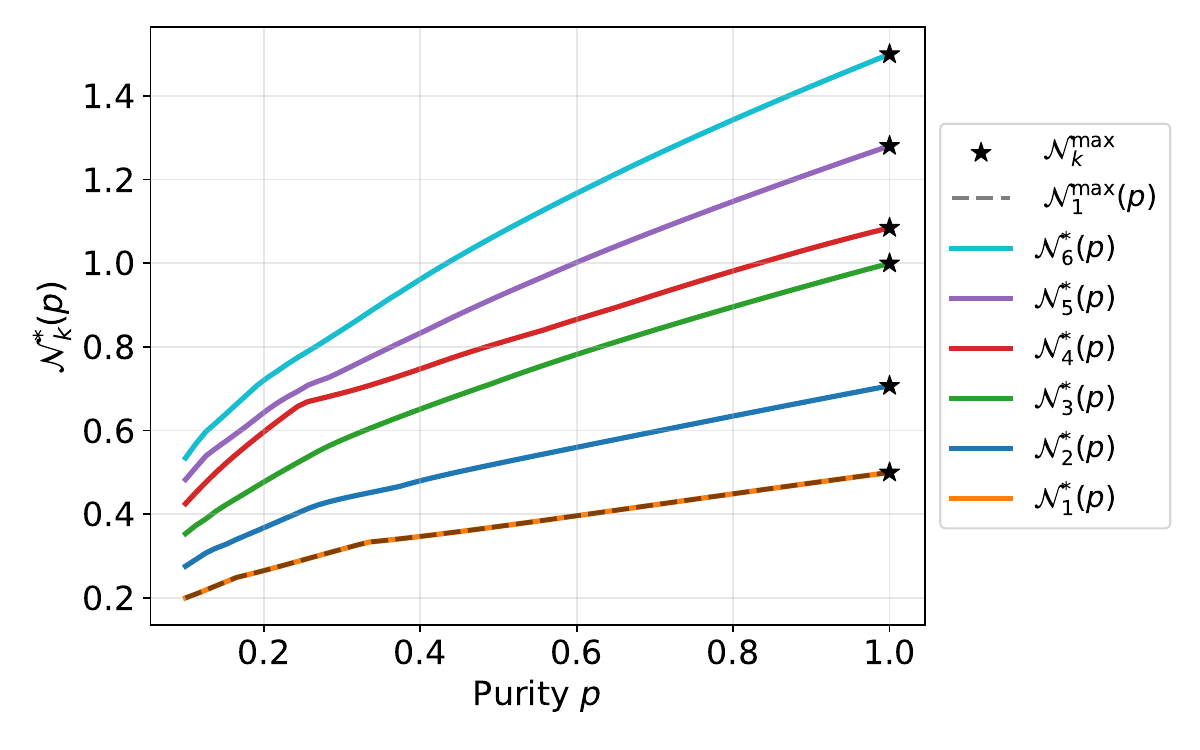}
    \caption {Estimated bounds $\cN_k^*(p)$ of Ky Fan $k$-negativity for $k=1,2,\dots,6$.
    The solid curves represent the numerical bounds obtained via
    seesaw algorithm. The numerical curve for $k=1$ coincides exactly with the
    analytical bound $\cN_1^{\max}(p)$ (dashed line). At $p=1$, the discrete
    markers indicate the exact bounds $\cN_k^{\max}$ derived from spectral graph
    theory.
    }
    \label{fig:purity_phase_diagram}
\end{figure}

Whether the optimal states for $k>1$ exhibit a recognizable structural pattern such as a generalization of structures found analytically for pure states and for $k=1$ is an interesting question that we leave for future investigation.

\section{Applications}\label{sec:app}

A core advantage of the Ky Fan negativity lies in how it differs from 
conventional entanglement measures: it does not always decrease under general 
LOCC operations. This non-monotonic behavior allows it to align with 
complex entanglement properties that do not follow standard LOCC ordering. 
A prime example is the so-called genuine multilevel entanglement
\cite{KraftCharacterizingGenuineMultilevel2018}, 
which is also closely related to the notion of irreducible quantum dimensions
\cite{CongWitnessingIrreducibleDimension2017}.

The notion of genuine multilevel entanglement concerns whether a high-dimensional entangled state can be simulated by multiple copies of low-dimensional entangled states. Consider a bipartite system with local dimension $d = d_1 \times d_2 \times \cdots \times d_m$, corresponding to a factorization of each party into $m$ subsystems of dimensions $d_1, d_2, \dots, d_m$, respectively. 
A pure state $\ket{\psi}$ is called decomposable with respect to this factorization if there exist local unitaries $U_A, U_B$ such that
\begin{equation}
     \ket{\psi} = (U_A \otimes U_B)\ket{\phi_1} \otimes \ket{\phi_2} \otimes \cdots \otimes \ket{\phi_m},
\end{equation}
where each $\ket{\phi_i}$ is a $d_i \times d_i$ entangled state. For mixed states,
a state $\rho$ is decomposable if it can be written as a convex combination of decomposable pure states. The set of all decomposable states is convex by definition.
A state that is not decomposable is called genuinely multilevel entangled.

A natural application of the Ky Fan $k$-negativity is to detect genuine multilevel entanglement. A key observation is that the Ky Fan $k$-negativity is invariant under local unitary transformations. Indeed, for $\rho' = (U_A \otimes U_B) \rho(U_A^\dagger \otimes U_B^\dagger)$, the partial transpose transforms as
$\rho'^{\Gamma} = (U_A \otimes U_B^*) \rho^{\Gamma} (U_A^\dagger \otimes U_B^T)$.
Since a unitary transformation preserves the spectrum of
$\rho^{\Gamma}$, $\cN_k(\rho') = \cN_k(\rho)$.
This invariance implies that, when computing the maximum of $\cN_k$ over all decomposable pure states, we can restrict product states $\ket{\phi_i}$ in the Schmidt basis without loss of generality.

The certification protocol then works as follows. For a given factorization
$d = d_1 \times d_2 \times \cdots \times d_m$, one first computes the maximum
Ky Fan $k$-negativity attainable by any decomposable state, which we denote by
$\cN_k^{\mathrm{dec}}$. $\cN_k^{\mathrm{dec}}$ can be obtained by optimizing over
the Schmidt coefficients of each state $\ket{\phi_i}$. If a state $\rho$ is prepared
in this system and the observed $\cN_k(\rho)$ strictly exceeds
$\cN_k^{\mathrm{dec}}$, then $\rho$ cannot be decomposable and is therefore
genuinely multilevel entangled. By the convexity of $\cN_k$, this criterion
holds for any state without requiring prior knowledge of the purity.
		
As a concrete example, consider the smallest nontrivial case $d=2\times 2$;
the generalization to the higher-dimensional and multi-party cases is straightforward.
Now, each party consists of two qubits, $A=A_1A_2$ and $B=B_1B_2$.
A decomposable pure state is a product of two two-qubit entangled states,
\begin{equation}
    \ket{\psi} = \ket{\varphi_1}\otimes\ket{\varphi_2},
\end{equation}
with Schmidt-decomposed states $\ket{\varphi_1}=x_0\ket{00}+x_1\ket{11}$ and
$\ket{\varphi_2}=y_0\ket{00}+y_1\ket{11}$ where $x_0\ge x_1\ge 0$, $y_0\ge y_1\ge 0$.
The six negative eigenvalues of $\ketbra{\psi}^{\Gamma}$ in absolute value are
\begin{align}
    \label{eq:product_neg}
    \begin{split}
        x_0 x_1 y_0^2, \quad x_0^2 y_0 y_1, \quad x_0 x_1 y_0 y_1, \\
        x_1^2 y_0 y_1, \quad x_0 x_1 y_1^2, \quad x_0 x_1 y_0 y_1.
    \end{split}
\end{align}
To see how the bound is obtained, consider $k=2$. By the ordering $x_0 \ge x_1$ and $y_0 \ge y_1$, the two largest eigenvalues in Eq.~\eqref{eq:product_neg} are $x_0 x_1 y_0^2$ and $x_0^2 y_0 y_1$, so
\begin{equation}
    \cN_2^{\mathrm{dec}} = \max_{\substack{x_0^2+x_1^2=1\\y_0^2+y_1^2=1}} (x_0 x_1 y_0^2 + x_0^2 y_0 y_1).
\end{equation}
The solution yields $\cN_2^{\mathrm{dec}} = 3\sqrt{3}/8$.
Maximizing $\cN_k$ over all such product states yields the bound summarized in Table~\ref{tab:product_bounds}.

\begin{table}
    \centering
    \caption{Maximum Ky Fan $k$-negativity for a product of two two-qubit entangled states, compared with the bound $\cN_k^{\max}$.}
    \label{tab:product_bounds}
    \begin{tabular}{c|c|c|c}
        \hline\hline
        $k$ & Product bound $\cN_k^{\mathrm{dec}}$ & Exact bound $\cN_k^{\max}$ & Tight? \\
        \hline
        $1$ & $1/2$ & $1/2$ & Yes \\
        $2$ & $3\sqrt{3}/8 \approx 0.650$ & $\sqrt{2}/2 \approx 0.707$ & No \\
        $3$ & $\approx 0.852$ & $1$ & No \\
        $4$ & $\approx 1.075$ & $\approx 1.085$ & No \\
        $5$ & $\approx 1.281$ & $\approx 1.281$ & Yes \\
        $6$ & $3/2$ & $3/2$ & Yes \\
        \hline\hline
    \end{tabular}
\end{table}

The strict gap between $\cN_k^{\mathrm{dec}}$ and
$\cN_k^{\max}$ for $k = 2, 3, 4$ yields direct certification criteria:
if $\cN_2(\rho) > 3\sqrt{3}/8$, $\cN_3(\rho) > 0.853$, or
$\cN_4(\rho) > 1.076$, then $\rho$ is genuinely multilevel entangled.

The second application we present is the quantitative estimation of
negativity. Recent works have devoted considerable attention to the
certification of entanglement from moments of the partial transpose \cite{ElbenMixedStateEntanglement2020,YuOptimalEntanglementCertification2021,NevenSymmetryresolvedentanglement2021}.
Among the resulting criteria, the $p_3$-PPT condition plays a central role.
Specifically, any separable state must satisfy
\begin{equation}
    p_3\ge p_2^2,
    \label{eq:p3ppt}
\end{equation}
where $p_k=\Tr[(\rho^\Gamma)^k]$. However, this condition provides only
qualitative certification. A violation of Eq.~\eqref{eq:p3ppt} certifies
the presence of entanglement but does not quantify its amount.

By taking advantage of the Ky Fan negativity, we obtain the following
lower bound on the negativity:
\begin{equation}
    \cN(\rho)\ge
    \frac{p_2^2-p_3}
    {f(p_2,p_3,\cN_1(\rho))},
    \label{eq:p3ppt_quant_pre}
\end{equation}
where $f$ is a positive function whose exact
expression is given in Appendix~\ref{app:negativity_bound}.
It is easy to see that any violation
of the $p_3$-PPT condition gives a nonzero lower bound on $\cN(\rho)$.
Moreover, $\cN_1(\rho)$ in Eq.~\eqref{eq:p3ppt_quant_pre} can be replaced by
any upper bound on it. In particular, using
$\Tr(\rho^2)=\Tr[(\rho^\Gamma)^2]$, we may replace
$\cN_1(\rho)$ with $\cN_1^{\max}(p_2)$, obtaining
\begin{equation}
\cN(\rho)\ge
\frac{p_2^2-p_3}
{f\bigl(p_2,p_3,\cN_1^{\max}(p_2)\bigr)}.
\label{eq:p3ppt_quant}
\end{equation}
It upgrades the $p_3$-PPT condition, which only provides qualitative
certification, to a quantitative lower bound on the negativity without
introducing additional moment information.
This feature also makes our criterion fundamentally different from the
results in Refs.~\cite{TarabungaQuantifyingmixedstate2026,
MillerDetectingentanglementfew2026}, which require higher-order or
absolute moments of the partial transpose and are therefore more
experimentally demanding.
Combined with our spectral bounds on the partial transpose, such
additional information can also be exploited to derive tighter lower
bounds on the negativity. This, however, lies beyond the scope of the
present work and will be reported elsewhere.

\section{Conclusion and Discussion}\label{sec:conclusion}

In this work, we have introduced the Ky Fan $k$-negativity to study
spectral properties of the partial transpose.
A common feature of the analytically solvable cases
is that the underlying high-dimensional optimization collapses to a
highly structured extremal configuration governed by only a few
effective parameters. For pure states, the optimal state has at most three
distinct Schmidt coefficients arranged in a core-spike structure, whose
spectral properties can be analyzed using classical results from graph
theory. For mixed states with $k=1$, the optimality conditions similarly
reduce the problem to a complete graph with only two distinct weights.
These structural reductions are precisely what make the exact solutions
possible. For mixed states with $k\ge 2$, the problem does not appear to
admit an analogous structural reduction.
We therefore complement numerical optimization with
dimension-independent analytical upper bounds and characterize their
scaling and asymptotic behavior.

Because these bounds require only the purity or a few measurable
moments of the partial transpose, they can be applied without full
state tomography. Such quantities are experimentally accessible
through randomized measurements
\cite{ElbenMixedStateEntanglement2020,
HuangPredictingmanyproperties2020}. This makes our results directly
applicable to the certification of quantum information processing
systems. As two concrete applications, we show that the Ky Fan
$k$-negativity robustly certifies genuine multilevel entanglement and
upgrades the $p_3$-PPT condition from a qualitative entanglement
criterion to a quantitative lower bound on the negativity.

Several questions remain open. The most intriguing is whether the
purity-constrained optimization for general $k$ admits a structural
reduction analogous to those found in the exactly solvable cases. The
patterns observed in the seesaw solutions may provide useful clues
toward such a reduction. Short of a complete solution, a more
accessible yet valuable goal is to determine $g_k$ exactly for finite
$k$. In addition, clarifying the operational meaning of the Ky Fan
$k$-negativity is an important task. In view of the role of the partial
transpose in entanglement distillation, it would be particularly
interesting to determine whether $\cN_k$ constrains the distillable
entanglement or the efficiency of distillation protocols. More
generally, we expect that the present framework can be extended to
constraints beyond the purity
\cite{ElbenRenyiEntropiesRandom2018,
ElbenMixedStateEntanglement2020,
YuOptimalEntanglementCertification2021}
and to multipartite systems
\cite{JungnitschTamingMultiparticleEntanglement2011}.

\begin{acknowledgments}
This work was supported by
the National Natural Science Foundation of China
(Grants No.~12574537, No.~12205170, and No.~12174224),
Shandong Provincial Department of Science and
Technology under Project No.~TQ012025002,
and the Shandong Provincial Natural Science Foundation of China
(Grant No.~ZR2022QA084). 
\end{acknowledgments}

\appendix
\section{$\cN_k^{\max}$ can only be achieved by pure states}\label{app:pure_state_proof}

Suppose a state $\rho$ achieves the global maximum $\cN_k^{\max}$. By the
convexity of $\cN_k$, for any convex decomposition
$\rho=\sum_ip_i\ketbra{\psi_i}$ (with $p_i > 0$ and $\sum_i p_i = 1$), the
equality $\cN_k(\rho) = \cN_k^{\max}$ requires that every constituent
pure state $\ket{\psi_i}$ individually achieves $\mathcal{N}_k^{\max}$. As
established in our structural analysis, this requires all optimal pure states to
possess unique, strictly positive Schmidt coefficients $\{s_i\}$ dictated by
the quotient matrix of the optimal core-spike graph $G_k^*$.

We now prove that $\rho$ must be pure. Suppose, to the contrary, that
$\rho$ is mixed, with rank $r\ge2$. According to the
Gisin-Hughston-Jozsa-Wootters theorem, any pure state within the support of
$\rho$ can participate in a valid convex decomposition
\cite{GisinQuantumMeasurementsStochastic1984,
Hughstoncompleteclassificationquantum1993}.
Therefore, for any two orthogonal pure states $\ket{\psi_1}$ and $\ket{\psi_2}$
in the support of $\rho$, any normalized superposition 
\begin{equation}
	\ket{\phi(\theta, \varphi)} = \cos\theta\ket{\psi_1} + \sin\theta \ee^{\ii\varphi}\ket{\psi_2}
\end{equation}
must also be an optimal pure state sharing the exact same Schmidt coefficients
$\{s_i\}$. Consequently, the reduced density matrix $\rho_A(\theta, \varphi) =
\Tr_B(\ketbra{\phi(\theta, \varphi)}{\phi(\theta, \varphi)})$ must maintain a
strictly invariant eigenvalue spectrum for any $\theta$ and $\varphi$. To
evaluate this constraint algebraically, we let $\rho_{A,1} =
\Tr_B(\ketbra{\psi_1})$, $\rho_{A,2} = \Tr_B(\ketbra{\psi_2})$, and define the
cross-interference operator $X = \Tr_B(\ketbra{\psi_1}{\psi_2})$. The reduced
density matrix expands as:
\begin{equation}
\begin{aligned}
		\rho_A(\theta, \varphi)& = (\cos^2\theta \rho_{A,1} + \sin^2\theta \rho_{A,2})\\
		&+ \cos\theta\sin\theta \ee^{\ii\varphi} X^\dagger + \cos\theta\sin\theta \ee^{-\ii\varphi} X.
\end{aligned}
\end{equation}

Note that all its moments must be constants independent of $\theta$ and $\varphi$. For any positive integer $n \ge 1$:
\begin{equation}
	P_n = \Tr(\rho_A(\theta, \varphi)^n) = \sum_{j=1}^K s_j^{2n}.
\end{equation}
Expanding the $n$-th moment yields a polynomial in $\ee^{\ii\varphi}$:
\begin{equation}
	\Tr(\rho_A(\theta, \varphi)^n) = \sum_{m=-n}^{n} \Gamma_m(\theta) \ee^{\ii m\varphi} = P_n.
\end{equation}
Consider the highest-frequency component $m=n$. To form the $\ee^{\ii n\varphi}$ term from an $n$-fold matrix product, the term $\cos\theta\sin\theta \ee^{\ii\varphi} X^\dagger$ must be consecutively selected in every single multiplication. Thus, its coefficient is uniquely determined without any cross-term: $\Gamma_n(\theta) = \cos^n\theta\sin^n\theta \Tr((X^\dagger)^n)$.

Multiplying the polynomial by $\ee^{-\ii n\varphi}$ and integrating over the full phase period $\varphi \in [0, 2\pi]$ annihilates all lower-frequency terms and the constant $P_n$, leaving only:
\begin{equation}
\begin{aligned}
	&\frac{1}{2\pi} \int_{0}^{2\pi} \Tr(\rho_A(\theta, \varphi)^n) \ee^{-\ii n\varphi} d\varphi =\cos^n\theta\sin^n\theta \Tr((X^\dagger)^n)\\
	&=\frac{1}{2\pi} \int_{0}^{2\pi}P_n\ee^{-\ii n\varphi} d\varphi  = 0.
\end{aligned}
\end{equation}
Noting that the equality holds for any $\theta \in (0, \pi/2)$, we rigorously obtain:
\begin{equation}
	\Tr((X^\dagger)^n) = 0 \quad \FA n \ge 1.
\end{equation}

By Newton's identities, if all moments of a matrix vanish, then all of
its eigenvalues must be zero. This indicates that
$X^\dagger$ (and equivalently $X$) is a nilpotent matrix. Consequently, its
determinant evaluated on any invariant subspace must be identically zero.
	
We now expose the ultimate algebraic contradiction by explicitly evaluating $X$
in the Schmidt basis of $\ket{\psi_1}$. 
We can express $\ket{\psi_1}$ in its Schmidt decomposition:
\begin{equation}
	\ket{\psi_1} = \sum_{i=1}^K s_i \ket{i}_A \ket{i}_B.
\end{equation}
Because the mixed state $\rho$ achieves the absolute maximum $\cN_k^{\max}$,
there exists an optimal Fantope projector $P^*$ such that $\Tr(-P^*
\rho^\Gamma) = \cN_k^{\max}$. By convexity, both pure states $\ket{\psi_1}$ and
$\ket{\psi_2}$ in the support of $\rho$ must independently saturate this maximum
using this identical projector. 
Define the local support of $P^*$
\begin{equation}
    \cP_A:=\supp\Tr_B(P^*),
    \qquad
    \cP_B:=\supp\Tr_A(P^*).
\end{equation}
In the Schmidt basis of $\ket{\psi_1}$, the optimal projector takes the form
\begin{equation}
    P^*=\sum_{(i,j)\in E(G_k^*)}
    \ketbra*{\psi_{ij}^-},
    \qquad
    \ket*{\psi_{ij}^-}=\frac{\ket{ij}-\ket{ji}}{\sqrt{2}}.
\end{equation}
Since the optimal graph $G_k^*$ is connected, the local supports of $P^*$ coincide with the Schmidt subspaces of $\ket{\psi_1}$,
\begin{equation}
    \cP_A=\Sp\{\ket{1}_A,\dots,\ket{K}_A\},
    \cP_B=\Sp\{\ket{1}_B,\dots,\ket{K}_B\}.
\end{equation}
Since $P^*$ is also optimal for $\ket{\psi_2}$, its local supports
likewise coincide with the Schmidt subspaces of $\ket{\psi_2}$. 
It implies we can expand $\ket{\psi_2}$ within the same basis:
\begin{equation}
	\ket{\psi_2} = \sum_{i,j=1}^K C_{ij} \ket{i}_A \ket{j}_B,
\end{equation}
for some $K \times K$ complex coefficient matrix $C$. Because $\ket{\psi_2}$ must possess the identical positive Schmidt spectrum $\{s_i\}_{i=1}^K$, the matrix $C$ must have singular values exactly equal to $s_i$. This guarantees that $C$ is a full-rank, strictly invertible matrix.

We directly compute $X = \Tr_B(\ketbra{\psi_1}{\psi_2})$:
\begin{equation}
	X = \sum_{i,j=1}^K s_i C_{ji}^* \ket{i}_A\bra{j}_A = S C^\dagger,
\end{equation}
where $S = \diag(s_1, s_2, \dots, s_K)$ is the diagonal matrix of the optimal Schmidt coefficients. 
Since all $s_i > 0$, the diagonal matrix $S$ is invertible. Since $C$
is invertible, its conjugate transpose $C^\dagger$ is also invertible.
Thus, $X = S C^\dagger$, being the product of two invertible matrices, must
be an invertible matrix within this subspace. 
However, an invertible matrix cannot be a nilpotent matrix. This implies no
such pure state $\ket{\psi_2}$ can exist in the support of $\rho$. We conclude
that the maximum of $\cN_k$ is exclusively attained on pure states.

\section{Proof of Lemma~\ref{lem:structure_mixed}}\label{app:proof_structure}
In this appendix, we give the detailed proof of Lemma~\ref{lem:structure_mixed}.
In addition to Eq.~\eqref{eq:elementwise}, we shall use the stationarity condition with respect to $s_i$. Since $M$ is symmetric, differentiating the Lagrangian with respect to $s_i$ gives
\begin{equation}
    2\sum_{j\neq i}M_{ij}s_j-2\lambda s_i=0,
\end{equation}
and hence
\begin{equation}\label{eq:eigen_equation}
    \lambda s_i=\sum_{j\neq i}M_{ij}s_j.
\end{equation}
Equivalently, $M\bs=\lambda\bs$, implying $\lambda$ is the eigenvalue for Perron eigenvector $\bs$. We first consider the non-degenerate case $\beta\neq0$, for which the structure is determined by the sign of $b$:

\textbf{Case 1: $b \ge 0$.}
Here, $\alpha \le 0$, then all non-zero vertices have $s_i > 0$, and $s_i s_j > 0$ is always satisfied. It implies that the graph is a complete graph.	
Using Eq.~\eqref{eq:eigen_equation}, the eigenvalue equation for each vertex $i$ becomes
\begin{equation}\label{eq:eigen_complete}
    \lambda s_i = \sum_{j \neq i} M_{ij} s_j = a s_i \left( \sum_{j \neq i} s_j^2 \right) + b \left( \sum_{j \neq i} s_j \right).
\end{equation}
Let $S_1 = \sum_{j} s_j$ be the sum of all components. Using the eigenvector normalization constraint $\sum_{j} s_j^2 = 1$, we can express the excluded sums as
\begin{equation}
    \sum_{j \neq i} s_j^2 = 1 - s_i^2 \quad \text{and} \quad \sum_{j \neq i} s_j = S_1 - s_i.
\end{equation}
Substituting these into Eq.~\eqref{eq:eigen_complete}, we get a universal cubic polynomial equation
\begin{equation}\label{eq:cubic polynomial}
    P(s_i) = -a s_i^3 + (a - b - \lambda)s_i + bS_1 = 0.
\end{equation}
For $b\ge0$, we have $P(0)=bS_1\ge0$ and $P(s)\to-\infty$ as $s\to\infty$. Moreover,
$P'(s)=-3as^2+(a-b-\lambda)$ has at most one positive root. Hence, $P(s)$ can have at most one positive root. Since every vertex value in the graph must be a root of this identical polynomial, the non-zero components of the optimal eigenvector $\bs$ must therefore be equal.

\textbf{Case 2: $b < 0$.}
We show that in this case the lowest-value class contains a single vertex and that the non-zero components can take at most two distinct values. Assume that the eigenvector has $m \ge 2$ distinct positive values. We arrange all vertex sets in descending order of their component values: $\cV_1, \cV_2, \dots, \cV_{m-1}, \cV_m$. Let $v_k$ denote the value of vertices in $\cV_k$, so $v_1 > v_2 > \dots > v_{m-1} > v_m > 0$.
We divide the analysis into two sub-cases according to the internal connectivity of the lowest class $\cV_m$:

\textbf{Sub-case 2.1: $\cV_m$ is internally disconnected ($v_m^2 \le \alpha$).}
We will show that $\cV_m$ must connect to the set directly above it, $\cV_{m-1}$.
For $m=2$, this follows immediately: otherwise, since $\cV_m$ is internally disconnected, every vertex in $\cV_m$ would be isolated. Then Eq.~\eqref{eq:eigen_equation} leads to $\lambda v_m=0$, contradicting $v_m>0$.
For $m\ge3$, assume $\cV_m$ does not connect to $\cV_{m-1}$ (i.e., $v_{m-1} v_m \le \alpha$). $\cV_m$ must connect to a higher-level vertex set $\cS$. Since $v_{m-1} > v_m$, the neighborhood of $\cV_{m-1}$ must contain $\cS$, plus a potential extra set $\Delta \cS$.
Using Eq.~\eqref{eq:eigen_equation} and Eq.~\eqref{eq:elementwise} on $v_m$ and $v_{m-1}$, we obtain
\begin{align}
    \lambda &= a \sum_{j \in \cS} {s_j}^2 + \frac{b}{v_m} \sum_{j \in \cS} s_j\\
    \lambda &= a \sum_{j \in \cS} {s_j}^2 + \frac{b}{v_{m-1}} \sum_{j \in \cS} s_j + \sum_{j \in \Delta \cS} \frac{s_j}{v_{m-1}} (a v_{m-1} s_j + b)
\end{align}
Equating them yields
\begin{equation}
    b \left( \frac{1}{v_m} - \frac{1}{v_{m-1}} \right) \sum_{j \in \cS} s_j = \sum_{j \in \Delta \cS} \frac{s_j}{v_{m-1}} (a v_{m-1} s_j + b).
\end{equation}
Since $b < 0$ and $v_{m-1} > v_m$, the left-hand side is strictly negative. However, since $\Delta \cS$ connects to $v_{m-1}$, the edge weight $(a v_{m-1} s_j + b) > 0$, making the right-hand side non-negative. This contradiction shows that $\cV_{m-1}$ must be
connected to $\cV_m$, i.e., $v_{m-1}v_m>\alpha$.
Since $\cV_{m-1}$ connects to $\cV_m$, and $v_1 > \dots > v_{m-1}$, it follows that all upper sets $\cV_1, \dots, \cV_{m-1}$ connect to $\cV_m$. Let the union of all these upper sets be $\cU$. Thus, the neighborhood of $\cV_m$ is exactly $\cU$. Moreover, any two vertices in $\cU$ have values at least $v_{m-1}$, and therefore $v_{m-1}^2>v_{m-1}v_m>\alpha$. Hence, $\cU$ forms a complete subgraph.
We next show that the lowest class $\cV_m$ contains exactly one vertex.

Assume $|\cV_m|\ge2$. We first exclude the exceptional case $|\cU|=1$. Let $r:=|\cV_m|\ge2$. Then the graph is a weighted star $K_{1,r}$, and all its edge weights $w_i$ are equal to a constant $c$.
The normalization and purity constraints give
\begin{equation}
    c=\frac{1}{r},
    \qquad
    p=\frac{1}{r},
\end{equation}
so its spectral radius is
\begin{equation}
    \lambda_0=\frac{1}{\sqrt r}.
\end{equation}
We show that this configuration is not locally optimal. Choose two vertices in $\cV_m$, and change their two star-edge weights to
\begin{equation}
    w_1=c+q,
    \qquad
    w_2=c-q-t,
\end{equation}
while keeping $w_i=c$ for $i=3,\ldots,r$, and add an edge of weight $t>0$ between them, where
\begin{equation}
    q=\frac{\sqrt{4t/r-3t^2}-t}{2}.
\end{equation}
For sufficiently small $t>0$, all weights remain positive. Moreover,
$q^2+qt+t^2=t/r$, and hence both the total weight and the total squared weight remain unchanged.
Let $M_t$ denote the resulting weighted adjacency matrix. If the newly
added edge is omitted, the remaining weighted star has spectral radius
\begin{equation}
    W=\sqrt{\frac{1}{r}-t^2}.
\end{equation}
Using its normalized Perron vector 
\begin{equation}
    \bm{u} = \frac{1}{\sqrt{2}}
    \left(1,\frac{w_1}{W},\frac{w_2}{W},\dots,\frac{w_r}{W}
    \right)^T
\end{equation}
as a test vector for $M_t$ gives
\begin{equation}
    \lambda_{\max}(M_t)\ge \bm{u}^TM_t\bm{u}
    =\sqrt{\frac{1}{r}-t^2}
    +\frac{t(1/r-t)^2}{1/r-t^2}.
\end{equation}
The right-hand side equals $1/\sqrt r$ at $t=0$ and has derivative
$1/r>0$ there. Therefore, for all sufficiently small $t>0$,
\begin{equation}
    \lambda_{\max}(M_t)>\frac{1}{\sqrt r}=\lambda_0,
\end{equation}
contradicting local optimality. We therefore only consider $|\cU|\ge2$.

Let $N$ denote the total count of non-zero elements $M_{ij}$. We define:
\begin{align}
    S_2 &= \sum_{ M_{ij}>0} s_i s_j, \label{eq:S2_def} \\
    S_4 &= \sum_{ M_{ij}>0} s_i^2 s_j^2. \label{eq:S4_def}
\end{align}
For all non-zero elements, the structural equation $M_{ij} = a s_i s_j + b$ holds. Substituting this into the two global constraints yields:
\begin{align}
    \sum_{M_{ij}>0} M_{ij} &= a S_2 + b N = 2, \\
    \sum_{M_{ij}>0} M_{ij}^2 &= a^2 S_4 + 2ab S_2 + b^2 N = 2p.
\end{align}
The first equation gives $b = \frac{2 - a S_2}{N}$.
The objective function to maximize is $\lambda = \sum_{M_{ij}>0} M_{ij} s_i s_j = a S_4 + b S_2$. Multiplying by $a$ and substituting the squared sum constraint expansion:
\begin{equation}
\begin{aligned}
    a \lambda &= a^2 S_4 + a b S_2\\
    &= (2p - 2ab S_2 - b^2 N) + a b S_2\\
    &= 2p - b(a S_2 + b N).
\end{aligned}
\end{equation}
Applying the sum constraint $(a S_2 + b N = 2)$, we obtain the fundamental algebraic invariant
\begin{equation} \label{eq:lambda_invariant}
    \lambda = \frac{2p - 2b}{a} = \frac{2p N - 4}{a N} + \frac{2 S_2}{N}.
\end{equation}
To determine $a$, we substitute $b = \frac{2 - a S_2}{N}$ back into the squared sum constraint
\begin{equation}
    2p = a^2 S_4 + 2a S_2 \left( \frac{2 - a S_2}{N} \right) + N \left( \frac{2 - a S_2}{N} \right)^2.
\end{equation}
Multiplying by $N$ and expanding the terms yields
\begin{equation}
    2p N = a^2 S_4 N - a^2 S_2^2 + 4.
\end{equation} 
Let $w = \frac{2p N - 4}{N}$. Because $|\cU|\ge2$, the non-zero weights $M_{ij}=a s_i s_j+b$ are not all equal. Therefore, the Cauchy-Schwarz inequality is strict: $(\sum M_{ij})^2<N\sum M_{ij}^2$. Thus, $4<2pN$ and $w>0$.
Since we require $a > 0$, $a = \sqrt{\frac{w}{S_4 - S_2^2/N}}$.
Substituting $a$ into Eq.~\eqref{eq:lambda_invariant} expresses $\lambda$ completely as a function of the sums:
\begin{equation} \label{eq:lambda_S2_S4}
    \lambda(S_2, S_4) = \sqrt{w \left(S_4 - \frac{S_2^2}{N}\right)} + \frac{2 S_2}{N}.
\end{equation}
Taking the partial derivative of $\lambda$ with respect to $S_2$ while keeping $S_4$ constant:
\begin{equation} \label{eq:partial_lambda}
\begin{aligned}
    \left. \frac{\partial \lambda}{\partial S_2} \right|_{S_4} &= \frac{1}{2\sqrt{w(S_4 - S_2^2/N)}} \cdot w\left( -\frac{2 S_2}{N} \right) + \frac{2}{N}\\
    &= -a \frac{S_2}{N} + \frac{2}{N} = b.
\end{aligned}
\end{equation}
For $b < 0$, we have $\frac{\partial \lambda}{\partial S_2} < 0$.

Assume $|\cV_m| \ge 2$. Letting $S_\cU = \sum_{\xi \in \cU} s_{\xi}$ and $S_{\cV_m} = \sum_{\sigma \in \cV_m} s_\sigma$, we can decouple the sums:
\begin{align}
    S_4 &= \sum_{\xi \neq \eta \in \cU} s_{\xi}^2 s_{\eta}^2 + 2 \left(\sum_{\xi \in \cU} s_{\xi}^2\right) \left(\sum_{\sigma \in \cV_m} s_\sigma^2\right), \label{eq:S4_decoupled} \\
    S_2 &= \sum_{\xi \neq \eta \in \cU} s_\xi s_\eta + 2 S_\cU S_{\cV_m}. \label{eq:S2_decoupled}
\end{align}
Consider a restricted local perturbation strictly within $\cV_m$: we perturb the values $s_\sigma$ while keeping the entire upper set $\cU$ fixed, and keeping the local sum of squares $\sum_{\sigma \in \cV_m} s_\sigma^2 = C_1 $ constant. 
According to Eq.~\eqref{eq:S4_decoupled}, $S_4$ depends on $\cV_m$ exclusively through $C_1$. Therefore, under this perturbation, $S_4$ is an absolute constant ($\Delta S_4 = 0$). 
Conversely, Eq.~\eqref{eq:S2_decoupled} dictates that $S_2$ is a strictly monotonically increasing linear function of the local sum $S_{\cV_m}$. 
Based on the Cauchy-Schwarz inequality, subject to the constraint $\sum_{\sigma \in \cV_m} s_\sigma^2 = C_1$, the sum $S_{\cV_m}$ is maximized when all vertices in $\cV_m$ have the same value. If $|\cV_m|\ge2$, a perturbation preserving $\sum_{\sigma\in\cV_m}s_\sigma^2$ strictly decreases $S_{\cV_m}$ and hence decreases $S_2$, while leaving $S_4$ unchanged. Since $\left.\partial\lambda/\partial S_2\right|_{S_4}<0$, this strictly increases $\lambda$, contradicting optimality.
Therefore, $|\cV_m|=1$.
Since $\cU$ is complete and this vertex is connected to every vertex in $\cU$, the graph is complete. Consequently for any vertex $i$, Eq.~\eqref{eq:cubic polynomial} holds.
For $b<0$, the cubic polynomial can possess at most two positive real roots. Since every vertex value in the graph must be a root of this identical polynomial, the non-zero components of the optimal eigenvector $\bs$ can take at most two distinct values, i.e., $m=2$.

\textbf{Sub-case 2.2: $\cV_m$ is internally connected ($v_m^2 > \alpha$).}
Since $v_m$ is the global minimum value in the network, the product of any two non-zero vertices $s_i, s_j$ must satisfy $s_i s_j \ge v_m^2 > \alpha$. 
This implies that the graph is complete.
Consequently for any vertex $i$, Eq.~\eqref{eq:cubic polynomial} holds and it implies we have only two sets $\cV_1$ and $\cV_2$. Then we will show that the lower value set $\cV_2$ has exactly one vertex ($|\cV_2|=1$).

The sum $S_2 = \sum_{i \neq j} s_i s_j$ can be exactly expressed using the sum $S_1 = \sum_{i} s_i$:
\begin{equation}
    S_2 = S_1^2 - \sum_{i} s_i^2 = S_1^2 - 1.
\end{equation}
Consider a local perturbation involving three arbitrary vertices assigned values $x_1, x_2, x_3$ which do not all belong to the same class.
We restrict the perturbation such that their local sum $\tilde{S} = x_1 + x_2 + x_3 $ and their local sum of squares $C_2 = x_1^2 + x_2^2 + x_3^2 $ are held constant.  
Because $\tilde{S}$ is conserved, the global sum $S_1$ is strictly invariant. This ensures that the global sum $S_2$ remains an absolute constant under this 3-vertex perturbation. 
According to Eq.~\eqref{eq:lambda_S2_S4}, maximizing $\lambda$ at a constant $S_2$ requires maximizing $S_4$. For a complete graph, 
\begin{equation}
    S_4 = \sum_{i \neq j} s_i^2 s_j^2 = (\sum_i s_i^2)^2 - \sum_i s_i^4 = 1 - \sum_{i} s_i^4.
\end{equation}
Therefore, maximizing $S_4$ is equivalent to minimizing the local sum of fourth powers: $F = x_1^4 + x_2^4 + x_3^4$.

Let $\mu = \tilde{S} / 3 > 0$ be the local mean value. We express each value as a deviation from the mean: $x_i = \mu + \Delta_i$. By definition, $\sum_{i=1}^3 \Delta_i = 0$. The sum of squares is $\sum_{i=1}^3 \Delta_i^2 = C_2 - 3\mu^2 \equiv R$. 
Expanding the fourth power sum $F$:
\begin{equation}
    F = \sum_{i=1}^3 (\mu + \Delta_i)^4 = 3\mu^4 + 6\mu^2 R + \sum_{i=1}^3 \Delta_i^4 + 4\mu \sum_{i=1}^3 \Delta_i^3.
\end{equation}
For any three zero-sum variables, the identity $\sum \Delta_i^4 \equiv \frac{1}{2}(\sum \Delta_i^2)^2$ strictly holds. Thus, $\sum \Delta_i^4 = \frac{1}{2}(R)^2$, making it an absolute constant on this perturbation path. 
Furthermore, utilizing the identity $\sum \Delta_i^3 = 3\Delta_1 \Delta_2 \Delta_3$, $F$ depends only on the product of the deviations $\Delta_1 \Delta_2 \Delta_3$.
To minimize $F$ (and hence maximize $\lambda$), the product of the deviations $\Delta_1 \Delta_2 \Delta_3$ must be strictly negative. Because the three deviations sum to zero, a negative product dictates that exactly one deviation is negative and two deviations are positive. This means any optimal 3-vertex subset must configure itself into one smaller value and two larger values. Applying this universal rule to the entire graph, the lower set $\cV_2$ must have only one vertex.

Having completed the analysis for $\beta\neq0$, it remains to consider the degenerate case $\beta=0$. The stationarity condition Eq.~\eqref{eq:stationary_M} then reduces to $s_is_j=\alpha-\gamma_{ij}$, so that
$s_is_j=\alpha$ on every edge and $s_is_j\le\alpha$ otherwise.
Hence, if a vertex has value $s_i$, every adjacent vertex must have value $\alpha/s_i$, and the values therefore alternate between $s_i$ and $\alpha/s_i$ along any path in the graph. So at most two distinct positive values can occur.
If only one value occurs, the conclusion is immediate. Otherwise, denote the two distinct values by $s_1>s_2>0$. Since the two classes must be connected,
\begin{equation}
    s_1s_2=\alpha,
\end{equation}
and therefore
\begin{equation}
    s_1^2>\alpha,
    \qquad
    s_2^2<\alpha.
\end{equation}
Thus, there can be only one vertex with value $s_1$, while vertices with value $s_2$ cannot be mutually connected. The graph is consequently a star graph $K_{1,r}$, whose edge weights are all equal to $\lambda s_2/s_1$ by Eq.~\eqref{eq:eigen_equation} applied
at each leaf. For $r\ge2$, such a configuration is not optimal, as shown in Sub-case 2.1; for $r=1$, the two vertex values coincide. 

In all cases, we prove that the vertices take at most two distinct values, and whenever two distinct values occur, the lower value is assigned to exactly one vertex.
This completes the proof of Lemma~\ref{lem:structure_mixed}.

\section{Upper Bound for $\cN_1^{\max}(p)$}
\label{app:upper_bound}

In this appendix, we rigorously prove $f(p) = \frac{2}{1 + \sqrt{1 + 8p^{-1}}}$ is an upper bound of $\cN_1^{\max}(p)$ for any given purity $p \in (0, 1)$. 

Consider the continuous relaxation where the effective number of vertices is extended to a continuous real parameter $v > 2$. Let $m_1$ and $m_2$ be the continuous edge weights subject to $m_1 \ge m_2 > 0$. 
The largest eigenvalue of the corresponding quotient matrix $Q$ is given by:
\begin{equation}
    \Lambda (v)= \frac{1}{2}\left[ (v-2)m_1 + \sqrt{(v-2)^2 m_1^2 + 4(v-1)m_2^2} \right].
\end{equation}
The continuous constraints for normalization and purity $p$ are:
\begin{align}
    n_1 m_1 + n_2 m_2 &= 1,  \\
    n_1 m_1^2 + n_2 m_2^2 &= p, 
\end{align}
where $n_1 = \frac{(v-1)(v-2)}{2}$ and $n_2 = v-1$. 

Noting that $\cN_1^{\max}(p)\le \max_v \frac12 \Lambda(v)$, to prove the target upper bound $\cN_1^{\max}(p)\le f(p)$ it suffices to show that $\Lambda(v) \le 2f(p)$ for any valid $v$. Substituting the explicit form of $f(p)$, we have:
\begin{equation}
    \Lambda(v) \le \frac{4}{1 + \sqrt{1 + 8p^{-1}}}.
\end{equation}
Since $\Lambda(v) \in (0,1)$, we can rearrange this inequality by isolating $p$:
\begin{align}
    \sqrt{1 + 8p^{-1}} &\le \frac{4 - \Lambda(v)}{\Lambda(v)}, \notag \\
    1 + 8p^{-1} &\le \frac{16 - 8\Lambda(v) + \Lambda(v)^2}{\Lambda(v)^2}, \notag \\
    p^{-1} &\le \frac{2 - \Lambda(v)}{\Lambda(v)^2}. \notag
\end{align}
Taking the reciprocal of both sides, the inequality $\cN_1^{\max}(p) \le f(p)$ is exactly relaxed into verifying:
\begin{equation}
    p\ge \frac{\Lambda(v)^2}{2-\Lambda(v)}, \quad \forall \text{ valid } v.
\end{equation}

For simplicity we introduce variables $x = n_1 m_1$, $\sigma = v - 1 > 1$. The physical feasibility requires $0<x<1$. Under this transformation, the eigenvalue and purity are reformulated as:
\begin{align}
    \Lambda &= \frac{x + \sqrt{x^2 + \sigma (1-x)^2}}{\sigma}, \label{eq:lambda_macro} \\
    p &= \frac{2x^2}{\sigma(\sigma-1)} + \frac{(1-x)^2}{\sigma}. \label{eq:purity_macro}
\end{align}
By isolating $\sigma$ from Eq.~\eqref{eq:lambda_macro}, we obtain:
\begin{equation}
    \sigma = \frac{(1-x)^2 + 2\Lambda x}{\Lambda^2}.\label{eq:sigma_lambda}
\end{equation}

We now construct the function $\Delta p = p - \frac{\Lambda^2}{2-\Lambda}$. $\Delta p$ expands into:
\begin{equation}
    \Delta p = \frac{2x^2 + (1-x)^2(\sigma-1)}{\sigma(\sigma-1)} - \frac{\Lambda^2}{2-\Lambda}.
\end{equation}
By substituting Eq.~\eqref{eq:sigma_lambda}, we obtain the exact fractional form:
\begin{equation}
    \Delta p = \frac{\cC(x,\Lambda)}{\cD(x, \Lambda)}, \label{eq:delta_p}
\end{equation}
where the denominator is defined as:
\begin{equation}
\begin{aligned}
    \cD(x, \Lambda) =& (2-\Lambda) [x^2 + 2(\Lambda-1)x + 1 ] \\
    &\times[ x^2 + 2(\Lambda-1)x + 1 - \Lambda^2 ].
\end{aligned}
\end{equation}
Since $\sigma > 1$ and $\Lambda \in (0,1)$, it is guaranteed that the denominator $\mathcal{D}(x, \Lambda) = (2-\Lambda)\Lambda^4\sigma(\sigma-1)$ is strictly positive. The sign of $\Delta p$ is therefore entirely dictated by the numerator $\cC(x)$. Remarkably, $\cC(x)$ permits an analytical factorization:
\begin{equation}
    \cC(x) =\Lambda^2 (x + \Lambda - 1)^2  (1-x)  [ (1-x) + \Lambda(1+x) ].
\end{equation}
We can easily verify the non-negativity of each factor within the feasible domain. Since all factors are non-negative, $\cC(x) \ge 0$ holds. Therefore, the inequality $p \ge \frac{\Lambda^2}{2-\Lambda}$ holds. This completes the proof. 

\section{Scaling and asymptotic tightness of the upper bound}
\label{app:scaling}

Recall that
\begin{equation}\label{eqa:gk-definition}
    g_k:=\sup_{\substack{P^2=P=P^\dagger\\\operatorname{rank}(P)=k}}
    \Tr\!\left[(P^\Gamma)_-^2\right],
\end{equation}
where the supremum is taken over all finite bipartite dimensions.
In this appendix we prove that
\begin{equation}\label{eqa:gk-properties}
    \frac{1}{2}\le\frac{g_k}{k}\le1,
    \qquad
    \lim_{k\to\infty}\frac{g_k}{k}=1,
\end{equation}
and use them to show that the bound in Eq.~\eqref{eq:upper_bound}
scales as $\sqrt{kp}$. For fixed $k$, this scaling is optimal up to at most a
factor of $\sqrt{2}$ in the low-purity limit $p\to0$. We further show
that the bound becomes asymptotically tight in the low-purity regime
$p\le 1/k$ as $k\to\infty$.

\subsection{General bounds on $g_k$}

For any rank-$k$ projector $P$, the invariance of the Hilbert-Schmidt
norm under partial transpose gives
\begin{equation}\label{eqa:gk-upper}
    \Tr[(P^\Gamma)_-^2]\le\Tr[(P^\Gamma)^2]=\Tr(P^2)=k,
\end{equation}
and taking the supremum yields $g_k\le k$.

For the lower bound, let $\ket{\Phi_d}=d^{-1/2}\sum_{i=1}^d\ket{ii}$ be
the maximally entangled state and $\Phi_d=\ketbra{\Phi_d}$. Since $\Phi_d^\Gamma$ has
eigenvalue $-1/d$ on the antisymmetric subspace of dimension $d(d-1)/2$,
\begin{equation}\label{eqa:rank-one-lower}
    \Tr \left[(\Phi_d^\Gamma)_-^2\right]
    =\frac{d(d-1)}{2}\frac{1}{d^2}=\frac{d-1}{2d}.
\end{equation}
Taking the direct sum of $k$ copies of $\Phi_d$ supported on pairwise locally orthogonal tensor-product subspaces yields a rank-$k$ projector $Q_{k,d}$ whose
partial transpose remains block diagonal, so
\begin{equation}\label{eqa:direct-sum-lower}
    \Tr \left[(Q_{k,d}^\Gamma)_-^2\right]=k\,\frac{d-1}{2d}.
\end{equation}
Taking $d\to\infty$ and combining with Eq.~\eqref{eqa:gk-upper} gives
\begin{equation}\label{eqa:gk-two-sided}
    \frac{k}{2}\le g_k\le k.
\end{equation}

\subsection{Asymptotic behavior of $g_k$}

We now show that the upper bound $g_k\le k$ is asymptotically sharp by constructing an explicit family of projectors for which $g_k/k\to1$. Fix an
integer $q\ge1$ and set $k=2^q$ and $n=2q+2$. On
$\mathds{C}^k=(\mathds{C}^2)^{\otimes q}$, define
\begin{equation}\label{eqa:Clifford-generators}
    \begin{aligned}
        G_{2j-1}&=Z^{\otimes(j-1)}\otimes X\otimes\I^{\otimes(q-j)},\\
        G_{2j}&=Z^{\otimes(j-1)}\otimes Y\otimes\I^{\otimes(q-j)},\\
        G_{2q+1}&=Z^{\otimes q},
    \end{aligned}
\end{equation}
for $j=1,\dots,q$, where $X,Y,Z$ are the Pauli matrices. These are
Hermitian unitaries obeying the Clifford relations
\begin{equation}\label{eqa:Clifford-relations}
    G_aG_b+G_bG_a=2\delta_{ab}\I_k,
\end{equation}
and are also commonly referred to as Majorana
operators \cite{BravyiFermionicQuantumComputation2002}. Set
\begin{equation}
    U_0=\ii\I_k,\quad
    U_a=G_a \quad\text{for $a=1,\dots,2q+1$}.
\end{equation}
Each $U_a$ is unitary, $U_0^\dagger=-U_0$, 
$U_a^\dagger=U_a$ for $a\ge1$, and the $U_a$ with $a\ge1$ mutually
anticommute; hence
\begin{equation}\label{eqa:U-relation}
    U_bU_a^\dagger=-U_aU_b^\dagger\qquad(a\neq b).
\end{equation}

Define
\begin{equation}\label{eqa:Clifford-isometry}
    W=\frac{1}{\sqrt n}\sum_{a=0}^{n-1}U_a\otimes\ket{a}.
\end{equation}
Since $W^\dagger W=n^{-1}\sum_aU_a^\dagger U_a=\I_k$, $W$ is an isometry
and therefore $P_k:=WW^\dagger$ is a rank-$k$ projector,
\begin{equation}\label{eqa:Clifford-projector}
    P_k=\frac{1}{n}\sum_{a,b=0}^{n-1}U_aU_b^\dagger\otimes\ketbra{a}{b}.
\end{equation}
Taking the partial transpose on the second subsystem and using
Eq.~\eqref{eqa:U-relation} gives
\begin{equation}\label{eqa:Clifford-partial-transpose}
    \begin{aligned}
        P_k^\Gamma
        &=\frac{1}{n}\sum_{a,b}U_bU_a^\dagger\otimes\ketbra{a}{b}\\
        &=\frac{1}{n}\sum_a\I_k\otimes\ketbra{a}
        -\frac{1}{n}\sum_{a\neq b}U_aU_b^\dagger\otimes\ketbra{a}{b}\\
        &=\frac{2}{n}\,\I_{kn}-P_k.
    \end{aligned}
\end{equation}
This simple affine relation determines the spectrum of $P_k^\Gamma$ immediately. Since $P_k^\Gamma$ has the negative eigenvalue $2/n-1$ on the support of $P_k$ with multiplicity $k$, and the positive
eigenvalue $2/n$ on its orthogonal complement, the negative part
$(P_k^\Gamma)_-=(1-2/n)P_k$. Hence
\begin{equation}\label{eqa:Clifford-bound}
    \frac{g_k}{k}\ge\frac{1}{k}\Tr\!\left[(P_k^\Gamma)_-^2\right]
    =\left(1-\frac{2}{n}\right)^2=\left(\frac{q}{q+1}\right)^2.
\end{equation}
Together with $g_k/k\leq1$, this implies
\begin{equation}
    \lim_{q\to\infty}\frac{g_{2^q}}{2^q}=1.
    \label{eqa:Clifford-subsequence}
\end{equation}
To extend the result from the sequence $k=2^q$ to arbitrary $k$, we use the superadditivity of $g_k$. Because $g_r$ and $g_s$ are defined as
suprema, let $\varepsilon>0$ and choose a rank-$r$ projector $P_r$ and a rank-$s$
projector $P_s$, respectively, such that
\begin{equation}
    \Tr\left[(P_r^\Gamma)_-^2\right]>g_r-\varepsilon,\quad
    \Tr\left[(P_s^\Gamma)_-^2\right]>g_s-\varepsilon.
\end{equation}
Placing $P_r$ and $P_s$ on locally orthogonal tensor-product
subspaces gives
\begin{equation}
    (P_r\oplus P_s)^\Gamma=P_r^\Gamma\oplus P_s^\Gamma.
\end{equation}
Consequently,
\begin{equation}
    \begin{aligned}
        g_{r+s}
        &\ge\Tr\left[(P_r^\Gamma)_-^2\right]
        +\Tr\left[(P_s^\Gamma)_-^2\right]\\
        &>g_r+g_s-2\varepsilon.
    \end{aligned}
\end{equation}
Since $\varepsilon>0$ is arbitrary,
\begin{equation}
    g_{r+s}\ge g_r+g_s.
\end{equation}
Fekete's lemma~\cite{FeketeUberdieVerteilung1923} then gives
\begin{equation}
    \lim_{k\to\infty}\frac{g_k}{k}
    =\sup_{k\ge1}\frac{g_k}{k}.
\end{equation}
Equation~\eqref{eqa:Clifford-subsequence} shows that the supremum
is at least $1$, whereas Eq.~\eqref{eqa:gk-upper} shows that it is
at most $1$. Consequently,
\begin{equation}\label{eqa:gk-limit}
    \lim_{k\to\infty}\frac{g_k}{k}=1.
\end{equation}

\subsection{Optimality of the low-purity scaling}

A simple embedding already shows that the $\sqrt{kp}$ scaling is optimal up to a constant factor. Fix $k$ and let $0<p\le1/k$. For any state $\rho$ with $\Tr(\rho^2)=kp$, enlarge one side of the bipartition by a maximally mixed $k$-dimensional ancilla and define
\begin{equation}\label{eqa:mixed-embedding}
    \tilde\rho=\frac{\I_k}{k}\otimes\rho,
    \qquad
    \Tr(\tilde\rho^2)=\frac{1}{k}\Tr(\rho^2).
\end{equation}
Since $\tilde\rho^\Gamma=(\I_k/k)\otimes\rho^\Gamma$, the $k$ smallest
eigenvalues of $\tilde\rho^\Gamma$ are the $k$ copies of
$\lambda_{\min}(\rho^\Gamma)/k$, and therefore
\begin{equation}\label{eqa:embedding-identity}
    \cN_k(\tilde\rho)=\cN_1(\rho).
\end{equation}
Maximizing over $\rho$ with $\Tr(\rho^2)= kp$ gives 
\begin{equation}\label{eqa:max-reduction}
    \cN_k^{\max}(p)\ge\cN_1^{\max}(kp).
\end{equation}
By the exact solution of the case $k=1$ in Sec.~\ref{sec:mixed_k1},
$\cN_1^{\max}(kp)=[1-o(1)]\sqrt{kp/2}$ as $p\to 0$, and combining with
Theorem~\ref{thm:upper_bound} gives
\begin{equation}\label{eqa:low-purity-two-sided}
    \begin{aligned}
        \left[1-o(1)\right]\sqrt{\frac{kp}{2}}
        &\le\cN_k^{\max}(p)\\
        &\le\sqrt{g_kp}\le\sqrt{kp},
        \qquad p\to0.
    \end{aligned}
\end{equation}
Thus, the $\sqrt{kp}$ scaling of the bound in Eq.~\eqref{eq:upper_bound} is optimal up to at most a factor of $\sqrt{2}$ in the low-purity limit.

\subsection{Asymptotic tightness in the low-purity regime}

We finally show that the upper bound becomes asymptotically tight,
\begin{equation}
    \lim_{\substack{k\to\infty\\ p\to0}}
    \frac{\cN_k^{\max}(p)}{\sqrt{kp}}=1.
\end{equation}
More precisely, we will prove that
$\cN_k^{\max}(p)\sim\sqrt{kp}$ as $k\to\infty$ in the regime $p\le 1/k$. 
For $k=2^q$, let $P_k$ be the projector of
Eq.~\eqref{eqa:Clifford-projector} and set
\begin{equation}\label{eqa:normalized-Clifford-state}
    \sigma_k:=\frac{P_k}{k},
\end{equation}
which has purity $1/k$. By
Eq.~\eqref{eqa:Clifford-partial-transpose}, $\sigma_k^\Gamma$ has
exactly $k$ negative eigenvalues, each equal to $-(1-2/n)/k$, so
\begin{equation}\label{eqa:Clifford-state-k-negativity}
    \cN_k(\sigma_k)=1-\frac{2}{n}=\frac{q}{q+1},
    \qquad k=2^q.
\end{equation}
Since Theorem~\ref{thm:upper_bound} gives
$\cN_k(\sigma_k)\le\sqrt{g_k/k}$ at purity $p=1/k$,
\begin{equation}\label{eqa:Clifford-squeeze}
    \frac{q}{q+1}\le\sqrt{\frac{g_{2^q}}{2^q}}\le1,
\end{equation}
and both sides approach $1$ as $q\to\infty$.

The construction can be adapted to every positive integer $k$ by
taking locally orthogonal direct sums. Write the binary expansion
$k=\sum_{j=0}^m\epsilon_j2^j$ with $\epsilon_j\in\{0,1\}$ and $m=\lfloor\log_2k\rfloor$. For every $j\ge 1$ with $\epsilon_j=1$, take a copy of the
rank-$2^j$ projector $P_{2^j}$. If $\epsilon_0=1$, take a two-qubit
maximally entangled rank-one projector. Place these projectors on
pairwise locally orthogonal tensor-product subspaces, and denote
their direct sum by $\tilde P_k$. Then $\tilde P_k$ is a
rank-$k$ projector. Define
\begin{equation}\label{eqa:binary-state}
    \tilde\sigma_k:=\frac{\tilde P_k}{k}.
\end{equation}
As before, $\Tr(\tilde\sigma_k^2)=1/k$. Because partial transpose preserves the block structure, the negative eigenvalues are obtained by collecting those of the blocks. Hence,
\begin{equation}\label{eqa:binary-state-k-negativity}
    \cN_k(\tilde\sigma_k)=\frac{1}{k}\sum_{j=0}^m\epsilon_j2^ja_j,
\end{equation}
where $a_0=1/2$ and $a_j=j/(j+1)$ for $j\ge1$. With
$L=\lfloor m/2\rfloor$ and $1-a_j\le1/(j+1)$,
\begin{equation}\label{eqa:binary-state-bound}
    \begin{aligned}
        1-\cN_k(\tilde\sigma_k)
        &=\frac{1}{k}\sum_{j=0}^m\epsilon_j2^j(1-a_j)\\
        &\le\frac{1}{k}\sum_{j=0}^L2^j+\frac{1}{L+2}\frac{1}{k}\sum_{j=L+1}^m\epsilon_j2^j\\
        &\le2^{L+1-m}+\frac{1}{L+2}
    \end{aligned}
\end{equation}
The right-hand side vanishes as $k\to\infty$, implying
        \begin{equation}
                \cN_k(\tilde\sigma_k)\to 1.
                \label{eqa:binary-state-limit}
            \end{equation} 
Since $\tilde\sigma_k$ has purity $1/k$, Theorem~\ref{thm:upper_bound} implies
\begin{equation}
    \cN_k(\tilde\sigma_k)\leq
    \cN_k^{\max}(1/k)\leq\sqrt{\frac{g_k}{k}}\le 1.
    \label{eqa:asymptotic-tightness-squeeze}
\end{equation}
Combining this with Eq.~\eqref{eqa:gk-limit}, the squeeze theorem yields
\begin{equation}
    \frac{\cN_k^{\max}(1/k)}
    {\sqrt{g_k/k}}
    \to 1.
    \label{eqa:asymptotic-tightness}
\end{equation}

This construction at $p=1/k$ can be extended to the entire regime $0<p\le 1/k$. Let $\mu=kp\le1$. For any $0<\varepsilon<\mu$, define a normalized state $\sigma_{p,\varepsilon}$ with purity $p$ as
\begin{equation}
    \sigma_{p,\varepsilon}
    =
    (\sqrt{\mu-\varepsilon}\tilde{\sigma}_k)
    \oplus \rho_\varepsilon,
\end{equation}
where $\rho_\varepsilon$ is an unnormalized separable state supported
on an orthogonal subspace and chosen such that
\begin{equation}
    \Tr(\rho_\varepsilon)
    =1-\sqrt{\mu-\varepsilon},
    \qquad
    \Tr(\rho_\varepsilon^2)
    =\frac{\varepsilon}{k}.
\end{equation}
Such a state can always be constructed in sufficiently large dimension.
Since a separable state does not contribute to the negative spectrum, we have
\begin{equation}
    \cN_k(\sigma_{p,\varepsilon})
    =
    \sqrt{\mu-\varepsilon}\,
    \cN_k(\tilde{\sigma}_k).
\end{equation}
Taking $\varepsilon\to0$ yields
\begin{equation}\label{eq:app-improved-lower-bound}
    \cN_k^{\max}(p)\ge \sqrt{kp}\,\cN_k(\tilde{\sigma}_k),
    \qquad 0<p\le\frac1k.
\end{equation}
Therefore,
\begin{equation}
    \cN_k(\tilde{\sigma}_k)\le\frac{\cN_k^{\max}(p)}{\sqrt{kp}}
    \le \sqrt{\frac{g_k}{k}},
    \qquad 0<p\le\frac1k.
\end{equation}
Since
$\cN_k(\tilde{\sigma}_k)\to1$ and $g_k/k\to1$, the two bounds converge to the same limit. Therefore,
\begin{equation}
    \frac{\cN_k^{\max}(p)}{\sqrt{kp}}
    \to 1,
    \qquad k\to\infty,
\end{equation}
for $0<p\le1/k$. Then the upper bound
in Eq.~\eqref{eq:upper_bound} is asymptotically tight as $k\to\infty$ in the low-purity regime.

\section{Relaxed purity constraint}\label{app:purity_relaxation}
		
In this appendix, we prove that relaxing the purity constraint from $\Tr(\rho^2) = p$ to $\Tr(\rho^2) \le p$ does not alter the maximum attainable value of the Ky Fan $k$-negativity. It suffices to show that $\cN_k^{\max}(p)$ is a non-decreasing function of $p$.

Suppose $\rho$ is an arbitrary quantum state with purity $\Tr(\rho^2) = q < p$.
We express it in its spectral decomposition as
\begin{equation}
    \rho = \sum_{i=1}^n q_i \ketbra{\varphi_i}{\varphi_i}
\end{equation}
where $\sum_{i=1}^n q_i^2 = q$. Let $\bm{q} = (q_1,\dots,q_n)^T$ denote the corresponding probability vector.

Because $q < p \le 1$, we can always construct a more concentrated probability vector $\bm p = (p_1, \dots, p_n)^T$ such that it majorizes the vector $\bm q$ (denoted as $\bm p \succ \bm q$) and satisfies $\sum_{i=1}^n p_i^2 = p$.

By the Hardy-Littlewood-P\'olya theorem \cite{MarshallInequalitiesTheoryMajorization2011}, the relation $\bm p \succ \bm q$ implies that there exists a doubly stochastic matrix $D$ such that $\bm q = D\bm p$. Furthermore, according to the Birkhoff-von Neumann theorem \cite{HornMatrixanalysis2012}, any doubly stochastic matrix $D$ can be decomposed into a convex combination of permutation matrices $G_\ell$. Thus, we can write
\begin{equation}
    D = \sum_\ell w_\ell G_\ell,
\end{equation} 
where $\{w_\ell\}$ forms a valid probability distribution ($w_\ell \ge 0, \sum_\ell w_\ell = 1$).

Let us define vectors $\bm p^{\ell} = G_\ell \bm p$, whose components are given by $p_i^\ell = (G_\ell \bm p)_i$. We then construct a family of density matrices $\rho_\ell$ sharing the identical eigenbasis as $\rho$, but with the permuted eigenvalues:
\begin{equation}
    \rho_\ell = \sum_{i=1}^n p_i^\ell \ketbra{\varphi_i}{\varphi_i}.
\end{equation} 
By construction, each $\rho_\ell$ has the exact target purity $\Tr(\rho_\ell^2) = \sum_i (p_i^\ell)^2 = p$.

The original state $\rho$ can now be reconstructed as a convex combination of these states:
\begin{equation}
    \rho = \sum_\ell w_\ell \rho_\ell.
\end{equation}

Finally, leveraging the convexity of the Ky Fan $k$-negativity, we obtain:
\begin{equation}
    \cN_k(\rho) \le \sum_\ell w_\ell \cN_k(\rho_\ell) \le \max_\ell \cN_k(\rho_\ell).
\end{equation} 

Since every $\rho_\ell$ is a valid quantum state with purity $p$, its $k$-negativity cannot exceed the maximum $\cN_k^{\max}(p)$ attainable at that specific purity. Therefore, $\cN_k(\rho) \le \cN_k^{\max}(p)$. Consequently, $\cN_k^{\max}(p)$ is monotonically non-decreasing, and optimizing over the relaxed inequality $\Tr(\rho^2) \le p$ is equivalent to optimizing over the equality $\Tr(\rho^2) = p$.

\section{Quantitative $p_3$-PPT criterion}
\label{app:negativity_bound}

In this appendix, we derive the explicit expression for the function
$f(p_2,p_3,\cN_1(\rho))$ that appears in Eqs.~\eqref{eq:p3ppt_quant_pre}
and \eqref{eq:p3ppt_quant}.
Denote the absolute values of the negative eigenvalues of $\rho^{\Gamma}$ in descending order by $\alpha_1,\alpha_2,\dots,\alpha_m$ and its non-negative eigenvalues by $\beta_1,\dots,\beta_n$.  Then $\cN(\rho)=\sum_{i=1}^m\alpha_i$, $\sum_{j=1}^n\beta_j=1+\cN(\rho)$, and the PT moments read
\begin{equation}
    p_2=\sum_{j=1}^n\beta_j^2+\sum_{i=1}^m\alpha_i^2,\qquad
    p_3=-\sum_{i=1}^m\alpha_i^3+\sum_{j=1}^n\beta_j^3.
\end{equation}

We consider the following vectors in $\mathds{R}^{n+m}$:
\begin{equation}
    \begin{aligned}
        u&=(\sqrt{\beta_1},\dots,\sqrt{\beta_n},\sqrt{\alpha_1},
        \dots,\sqrt{\alpha_m}),\\
        v&=(\beta_1^{3/2},\dots,\beta_n^{3/2},\alpha_1\sqrt{\alpha_1},
        \dots,\alpha_1\sqrt{\alpha_m}).
    \end{aligned}
\end{equation}
The Cauchy-Schwarz inequality gives
\begin{equation}
    \label{eq:gram_basic}
    \|u\|^2\|v\|^2\ge\bigl|\langle u,v\rangle\bigr|^2.
\end{equation}
By taking $\|u\|^2=1+2\cN(\rho)$, $\langle u,v\rangle \ge p_2$, and $\|v\|^2 \le p_3+2\alpha_1^2\cN(\rho)$, Eq.~\eqref{eq:gram_basic} yields the key inequality
\begin{equation}
    \label{eq:key_gram}
    (1+2\cN(\rho))(p_3+2\alpha_1^2\cN(\rho))\ge p_2^2.
\end{equation}
Expanding Eq.~\eqref{eq:key_gram} gives a quadratic inequality for $\cN(\rho)$:
\begin{equation}
    \label{eq:quad_gram}
    4\alpha_1^2 \cN(\rho)^2+2(\alpha_1^2+p_3)\cN(\rho)+(p_3-p_2^2)\ge 0.
\end{equation} 
Solving for $\cN(\rho)\ge 0$ yields
\begin{equation}
    \label{eq:gram_sol}
    \cN(\rho)\ge\frac{p_2^2-p_3}{\sqrt{(\alpha_1^2-p_3)^2+4\alpha_1^2p_2^2}+\alpha_1^2+p_3}.
\end{equation}
Noting that $\alpha_1=\cN_1(\rho)$, the function $f$ in Eq.~\eqref{eq:p3ppt_quant} is given by
\begin{equation}
    \begin{aligned}
        f(p_2,p_3,\cN_1(\rho))=&\sqrt{\bigl[\cN_1(\rho)^2-p_3\bigr]^2+4\cN_1(\rho)^2p_2^2}\\
        &+\cN_1(\rho)^2+p_3.
    \end{aligned}
\end{equation}	
Moreover, Eq.~\eqref{eq:key_gram} shows that
$\alpha_1=\cN_1(\rho)$ in Eq.~\eqref{eq:gram_sol} can be replaced by
any upper bound on $\cN_1(\rho)$. In particular, using
$\Tr(\rho^2)=\Tr[(\rho^\Gamma)^2]$, we may replace $\cN_1(\rho)$ with
$\cN_1^{\max}(p_2)$, obtaining the quantitative $p_3$-PPT criterion:
\begin{equation}
\cN(\rho)\ge
\frac{p_2^2-p_3}
{f\bigl(p_2,p_3,\cN_1^{\max}(p_2)\bigr)}.
\end{equation}

\bibliography{kNegativityRef}

\end{document}